\documentclass[10pt]{article} 
\usepackage[accepted]{tmlr}

\usepackage{amsmath,amsfonts,bm}

\def\eqref#1{equation~\ref{#1}}

\def\1{\bm{1}}

\DeclareMathAlphabet{\mathsfit}{\encodingdefault}{\sfdefault}{m}{sl}
\SetMathAlphabet{\mathsfit}{bold}{\encodingdefault}{\sfdefault}{bx}{n}

\usepackage{url}
\usepackage{graphicx}
\usepackage{wrapfig}

\usepackage{enumerate}
\usepackage{amsfonts}
\usepackage{textcomp}
\usepackage{amssymb}
\usepackage[linesnumbered,ruled,vlined]{algorithm2e}
\usepackage{relsize}
\usepackage{algpseudocode}
\usepackage{float}
\usepackage{multirow}
\usepackage{caption}

\usepackage[utf8]{inputenc} 
\usepackage[T1]{fontenc}    
\usepackage{subcaption}
\usepackage{amsfonts}
\usepackage{textcomp}
\usepackage{amssymb,amsthm}
\usepackage[linesnumbered,ruled,vlined]{algorithm2e}
\usepackage{hyperref}       
\usepackage{booktabs}       
\usepackage{nicefrac}       
\usepackage{microtype}      
\usepackage{xcolor}         
\usepackage{comment}
\usepackage[normalem]{ulem}
\newtheorem{theorem}{Theorem}[section]
\newtheorem{proposition}[theorem]{Proposition}
\newtheorem{lemma}[theorem]{Lemma}

\newtheorem{definition}[theorem]{Definition}
\newtheorem{assumption}[theorem]{Assumption}
\newtheorem{remark}[theorem]{Remark}
\newtheorem*{proposition1}{Proposition~\ref{prop: tail_unconf}}
\newtheorem*{proposition2}{Proposition~\ref{prop:identifiability}}
\newtheorem*{remark1}{Remark~\ref{remark:regular_varying}}
\newtheorem*{remark2}{Remark~\ref{remark:point_process}}

\usepackage{bbm}
\usepackage{stackengine}
\usepackage{graphicx}

\newcommand{\ind}{\perp\!\!\!\!\perp} 
\newcommand\barbelow[1]{\stackunder[1.2pt]{$#1$}{\rule{.8ex}{.075ex}}}

\title{Treatment Effects in Extreme Regimes}

\author{\name Ahmed Aloui \email ahmed.aloui@duke.edu \\
      \addr Department of Electrical and Computer Engineering\\
      Duke University
      \AND
      \name Ali Hasan \email ali.hasan@morganstanley.com \\
      \addr Machine Learning Research\\
      Morgan Stanley
      \AND
      \name Yuting Ng \email  yuting.ng@duke.edu\\
      \addr Department of Electrical and Computer Engineering\\
      Duke University
      \AND
      \name Miroslav Pajic \email miroslav.pajic@duke.edu\\
     \addr Department of Electrical and Computer Engineering\\
      Duke University
      \AND
      \name Vahid Tarokh \email vahid.tarokh@duke.edu \\
      \addr Department of Electrical and Computer Engineering\\
      Duke University
}

\def\month{10}  
\def\year{2026} 
\def\openreview{\url{https://openreview.net/forum?id=u8ZCcv0gTV}} 

\begin{document}

\maketitle

\begin{abstract}
Understanding the effects of treatment in extreme regimes is important in characterizing the risks associated with different interventions. However, both the lack of counterfactual outcomes and the difficulty of collecting extreme data pose significant challenges. This motivates us to propose a new framework based on the extreme value theory for estimating the treatment effects in extreme regimes. We quantify these effects using variations in the tail decay rates of distributions of potential outcomes, respectively, for the control and  treatment groups. We provide a theory and algorithms for calculating these quantities. 
Numerical results are provided using standard synthetic and semi-synthetic datasets validating our proposed framework.
\end{abstract}

\section{Introduction}
\label{sec:intro}
Understanding the effect of an intervention is a primary objective of causal inference. In the Neyman-Rubin framework~\citep{neyman1923,rubin1974}, the average treatment effect (ATE) and the conditional average treatment effect (CATE) are the quantities of general interest. These metrics are instrumental for informed decision-making ~\citep{abadie2018econometric} and have been widely applied to various areas such as econometrics and epidemiology~\citep{rothman2005causation}. 

Despite their importance, these statistics are based on the \emph{expectation} of the difference of the outcome variables influenced by the intervention. 
However, since they are focused on the expectation, these quantities do not consider the \textit{extreme} behavior of the treatment effect.
In many scenarios, it is equally or perhaps more important to understand the \emph{worst case} behavior of interventions rather than only the expected behavior. 
For example, we may prescribe a treatment that is beneficial in average, but may increase the likelihood of extremely adverse events that, a scenario that we want to avoid completely in some applications.  This is particularly relevant where an intervention can lead to severe adverse effects for a subset of the population. 
In medical settings, a specific patient may experience stronger adverse side effects in response to treatment than the vast majority of the population. 
In light of above, there may be a need to characterize \emph{worst case} (or similarly \emph{best case}) effects of the treatment rather than the merely the average. 

This provides the motivation for the present work, where we propose a framework for estimating \emph{tail} behavior of the treatment effects. 
To this end, we note that quantile methods have found application in causal inference~\citep{diaz2017efficient,giessing2021debiased}, but their performance becomes notably unstable and inaccurate for high quantiles, especially when the amount of data is limited. 
Thus, we instead eemploy Extreme Value Theory (EVT)~\citep{haan2006extreme} to propose a method to quantify extreme treatment effects.
As discussed above, EVT provides asymptotic theory for the behavior of the tail of a distribution using Generalized Extreme Value (GEV) distributions.
GEVs are characterized using their location, scale and shape parameters.  The latter is particularly important since it quantifies the rate of decay of the tail of the distribution.  This motivates us 
to define \emph{extreme treatment effect} (ETE)  as the difference between the shape parameters of the tails of the distributions of control and treatment outcomes. In other words,  ETE quantifies the change in the tail decay rate following the administration of the treatment. 
In some cases, we will also be interested in changes in the tail decay rates of the conditional potential outcomes. To this end, we will define \emph{conditional extreme treatment effect} (CETE) in an analogous manner.

\paragraph{Motivating Example}
In applications such as healthcare, it may be necessary to ensure that the tail behavior of a continuous severity outcome (e.g. a laboratory biomarker quantifying the severity of an adverse reaction to a drug) does not significantly change after prescribing a treatment (intervention). 
We can use ETE to determine whether the treatment is likely to produce significantly more extreme outcomes, since it quantifies the difference in the decay rates of the tails.  
Obviously, faster tail decay asymptotically imply smaller numbers of extreme events. 
The idea is illustrated in Figure~\ref{fig:ite_vs_itte} where only an outcome above $17$ corresponds to a serious adverse event.
The treatment seems to be beneficial in average (since it has a negative ATE). However, it corresponds to a greater risk of serious adverse effects (corresponding to a higher ETE). 
Additionally,  for some heavy tail distributions (for instance the Cauchy distribution) first order moments may not exist. In this light, mean-based estimators such as the ATE may not always be well-defined. In these cases, the proposed CETE and ETE remain well-defined and applicable, which is another advantage of the proposed framework.
An additional benefit of the proposed framework is the potential applications of CETE in fields such as personalized medicine because it can be applied by conditioning to small subsets of the population.

\begin{wrapfigure}{r}{3in}
\includegraphics[width=3in]{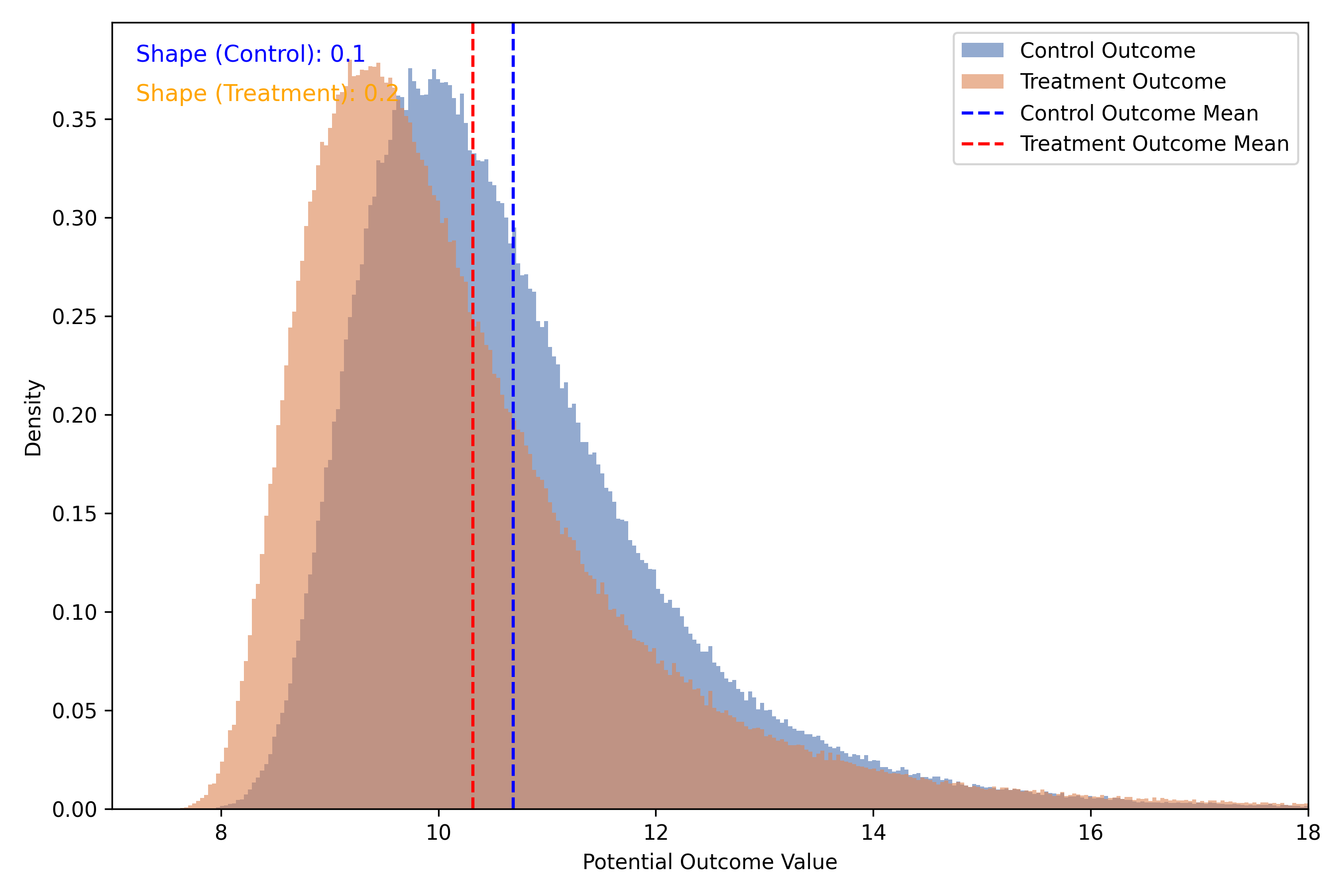}
\caption{ETE Versus ATE. The control group has a faster decaying tail (shape parameter $=0.1$), than that of the treatment group ($0.2$), even though, the treatment group has a lower mean. Thus ETE $=0.1 > 0$, and ATE ($\sim -0.37$ < 0).}
\label{fig:ite_vs_itte}
\end{wrapfigure}

We illustrate this idea in greater detail in Appendix~\ref{sec:motivation}.
Note that for the problems that we are studying, higher order moments such as variance or skew are insufficient for understanding the tail risk.
The reason is due to the fact that many of the underlying phenomena are heavy tailed, and higher order moments may not be defined for these data, resulting in statistics that never converge. 

While capturing essential characteristics of the potential outcomes, the proposed ETE and CETE have challenges associated with estimation from data. 
The first challenge is due to the lack of observed data in the tails of the distributions. 
By definition, data in the tails are rare, hence we would not expect to have access to many samples from the tail. 
To counteract this challenge, our framework uses EVT to represent the tail distribution as a parametric class of distributions and extrapolate deep into the tails from near-tail data. 
The second challenge arises from the unavailability of counterfactual outcomes, which is referred to as the \emph{fundamental problem of causal inference} ~\citep{rubin1974,holland1986}.
This challenge persists in any causal inference setting and is due to the fact that we cannot observe the outcome for the same individual under two different treatment regimes. 
The third challenge is due to \emph{unobserved confounders}. 
While it is common to assume that the observations are generated such that no unobserved confounders exist, which is known as the unconfoundedness assumption\footnote{Also referred to as conditional exchangeability.}, it is often very hard to test this assumption in practice. 

To address these challenges, we propose a new \emph{tail unconfoundedness} condition that we prove to be strictly weaker than unconfoundedness. Moreover, we develop statistical estimators for the proposed ETE and CETE and we prove identifiability and consistency of the CETE estimator. 
We illustrate these results numerically in various experiments with both synthetic and semi-synthetic datasets.   

\textbf{Related Work} 
As discussed in the introduction, EVT provides a mathematical theory understanding the tails of a distribution~\citep{haan2006extreme}. We refer to \citet{chavez2024causality} for a recent survey of the growing literature at the intersection of causality and extremes.
In causal inference, methods for estimating the Average Treatment Effect (ATE),  Conditional Average Treatment Effect (CATE), extreme average treatment effects, and quantile treatment effects have been explored in the literature.
Methods for estimating the ATE include covariate adjustment of~\citet{rubin1974}, also known as the backdoor adjustment in the DO-calculus framework~\citep{pearl}. Deep learning methods for estimating CATE have also been well-studied, most notably $s$-learner and $t$-learner neural networks~\citep{kunzel2019metalearners}, counterfactual regression~\citep{johanson}, and TARNet~\citep{shalit}. Other machine learning approaches such as random forests \citep{athey} and Gaussian processes \citep{alaa} have also been applied to CATE estimation. Yet another notable method is doubly robust learners for CATE estimation~\citep{kennedy2020towards}.

Estimators for extreme quantile treatment effects have also been developed in the literature in~\citep{zhang2018extreme,li2026tail}, but do not extrapolate well outside of the range of observations. The closest existing framework to our work is perhaps given by Deuber \citet{deuber2024estimation}. The aforementioned work estimates extremal quantile treatment effects for heavy-tailed distributions ith extrapolation beyond the data range via a causal Hill estimator; their target is a counterfactual quantile at a fixed extreme level. In contrast, our work is focused on changes in the tail decay rates, as discussed above.

Inferring the conditional value at risk for a treatment effect was investigated in~\citet{kallus2023treatment}, however tail quantiles were not explicitly considered. 
Estimating the extreme average treatment effect was considered in~\citep{huang2022extreme} using extreme value theory. Moreover, \cite{tan2026estimation} studies the expected potential outcomes in the extreme covariate regions. 

Outside the causal analysis literature, \citet{pasche2024neural} models covariate-dependent tail parameters $(\sigma(x), \xi(x))$ with neural networks in a peaks-over-threshold framework, and \citet{sasaki2022fixed} estimates conditional extreme quantiles from fixed numbers of nearest-neighbour tail observations.  The latter provides a construction related to our $\varepsilon$-max-sampler; neither addresses counterfactual estimands or selection bias. Works on estimating a parametric or semi-parametric form of the counterfactual distribution include \cite{kennedy2021semiparametric,chernozhukov2005iv,chernozhukov2013inference}. 
Methods for estimating the quantiles include \citep{diaz2017efficient,firpo2007efficient,frolich2013unconditional,zhang2012causal,giessing2021debiased}.  These methods do not consider estimating the conditional extreme treatment effect and focus on the population level. 

Our work departs from the above literature by providing a new framework that directly exploits the extrapolating behavior of EVT. Our approach is closely related to the study of max-stable processes, which have been applied in various disciplines to model extremes over function spaces~\citep{davison2012spatial}.  However, a key difference is that these  are often estimated assuming an abundance of temporal observations, which is an assumption that may not hold in causal inference settings.

\section{Background}
\subsection{Causal inference}
Let $\mathcal{X} \subset \mathbb{R}^d$ and consider covariates $X \in \mathcal{X}$.
Let $T \in \mathcal{T} = \{0, 1\}$ denote the treatment variable,  we refer to groups corresponding to $T=0$ and $T=1$ respectively as the control and treatment groups.  We denote factual outcome by outcomes of a random variable $Y\in \mathbb{R}$. 
Potential outcomes are the outcomes that would have been observed if only treatment $T=1$ (respectively $T=0$) was administered and will be denoted by $Y_1$ (respectively $Y_0$) throughout this work~\citep{morgan2015counterfactuals}.
Obviously, the relationship between the factual outcome and the potential outcomes can be written as $Y = Y_1 T + Y_0 (1-T)$. 

The main challenge in causal inference is that we do not observe the counterfactual outcomes, i.e. the outcomes that would have been observed if the corresponding decisions about applying/not applying the treatment were reversed. 

We refer to any statistics of the potential outcomes as a \emph{causal statistics}.  A causal statistic is said to be identifiable if it can be reduced to a function of the factual outcomes $Y$, the treatment assignment $T$, and the covariates $X$~\citep{imbens2004}. 
The following definition formalizes this concept:
\begin{definition}[Causal Identifiability]
    \label{def:ident}
    A causal statistic (e.g $\mathbb{E}[Y_1]$) is identifiable if it can be computed from a purely observational quantity (e.g. $ \mathbb{E}[\mathbb{E}[Y \mid T=1,X]]$).
\end{definition}
Similarly, an estimator, as a particular statistic, is said to be causally identifiable when it meets a similar criterion. 

We next describe some common assumptions that guarantee causal identifiability of CATE and ATE.
The first is consistency, which requires that the potential outcomes align with the factual outcome:
\begin{assumption}[Consistency]
\label{ass:consistency}
For $t\in \mathcal{T}$, the following holds, $T=t \implies Y_t = Y$.
\end{assumption}
The second is overlap\footnote{Also known as the positivity assumption}. This requires that the different treatment groups have the same support: 
\begin{assumption}[Overlap]
\label{ass:pos}
$$
\forall x \in \mathcal{X}, \; 0 < P(T=1\mid X=x) < 1 
$$
\end{assumption}


Finally, an unconfoundedness assumption is often made. This requires that the potential outcomes are independent of the treatment given the covariates:
\begin{assumption}[Conditional Unconfoundedness]
\label{ass:unconf}
    $\left(Y_1,Y_0\right)\ind T \mid X$
\end{assumption}

\subsection{Extreme value theory}
EVT provides a mathematical framework for extrapolating to the tails of a distribution. 
The Fisher-Tippett-Gnedenko Theorem states that if the shifted and scaled versions of maximum order statistics of a distribution $P(\cdot)$ converge to a non-degenerate distribution, then it is guaranteed to converge to a family of well-specified distributions known as the generalized extreme value (GEV) distributions~\citep{haan2006extreme,fisher1928,gnedenko1943}.
If this convergence holds, the distribution $P(\cdot)$ is said to be in the \emph{maximum domain of attraction (MDA)}\footnote{The MDA of a distribution $p$ is the distribution to which the shifted and scaled versions of maxima of samples from $p$ converges to per Fisher-Tippett-Gnedenko Theorem.} of the GEV.

A GEV distribution is characterized by three parameters: $\mu$, the location parameter; $\sigma$, the scale parameter; and, $\xi$ the \emph{shape} or \emph{tail index} parameter, more explicitly:
\begin{definition}
\label{def:gev}
The GEV distribution is defined as:
\begin{equation}
G_\xi(y)=\exp \left(-(1+\xi \bar{y})^{-1 / \xi}\right), \:\: 1+\xi \bar{y} > 0
\label{eq:gev}
\end{equation}
where 
$\bar{y} = \frac{y-\mu}{\sigma}$
with parameters $\mu, \sigma, \xi$
The case $\xi =0$ is defined as the limit where $\xi$ approaches zero.
\end{definition}
For our purposes the shape parameter $\xi$ is of interest, which governs the rate of tail decay.
When $\xi > 0$ the distribution exhibits a heavy tail, whereas when $\xi < 0$ the tail is bounded.
If $\xi = 0$, the tail decays exponentially.
The Fr\'echet, Weibull, and Gumbel distributions respectively correspond to positive, negative, and zero values of $\xi$. 
\begin{theorem}[Fisher–Tippett–Gnedenko]
    \label{thm:gev}
    Let $Y^{(1)}, Y^{(2)}, \ldots, Y^{(n)}$ be a sequence of iid real random variables. If there exist two sequences of real numbers $a_n>0$ and $b_n \in \mathbb{R}$ such that the following limits converge to a non-degenerate distribution function:
$$
\lim _{n \rightarrow \infty} P\left(\frac{\max \left\{Y^{(1)}, \ldots, Y^{(n)}\right\}-b_n}{a_n} \leq y\right)=G(y),
$$
then the limiting distribution $G$ is the GEV distribution.
\end{theorem}

In general, fitting such distributions requires many repeated observations such that the maximum can be computed, thus obtaining samples approaching the tail.
The \emph{block maxima} method refers to the approach when the data is subdivided into blocks and the maximum quantity per block is used for fitting~\citep{coles2001introduction}.
Specifically, $K \times n$ observations $\left\{Y^{(i,j)}\right\}_{i,j=1}^{K \times n}$ observations are divided into $K$ blocks of size $n$.
Let $M^{(i)}= \underset{j=1,\ldots, K}{\max} Y^{(i,j)}$ for $i = 1,\ldots,n$.
The values $\left\{M^{(i)}\right\}_{i=1}^{n}$ are then used to fit the GEV model.
In the causal inference setting, multiple observations per individual is often prohibitively expensive to obtain, so we introduce alternatives to the block maxima method later in the text.

\section{Theoretical Framework}
We now present the proposed framework for estimating treatment effects in extreme regimes.
We begin by defining the estimation problem motivated by the previous examples.
Then, we establish an identifiability theorem under the condition of \textit{tail unconfoundedness}, which is a weaker requirement than conditional unconfoundedness.
Combining these results leads to a numerical algorithm that can be easily implemented. 
The proofs of our results are presented in Appendix~\ref{sec:proofs}.

\subsection{Problem setup}

Let $(\Omega,\mathcal F,P)$ carry the covariate $X:\Omega\to\mathcal X$,
treatment $T$, and factual outcome $Y$. For each $x\in\mathcal X$ and
$t\in\{0,1\}$, we extend this space to also carry an i.i.d.\ sequence
$\big(Y_t^{(i)}(x)\big)_{i\ge1}$ with common law equal to that of $Y_t\mid X=x$. This sequence represents $n$ hypothetical replicate outcomes of a single individual with covariate $x$ under treatment $t$,
equivalently, we condition on the degenerate event $x_i=x_j=x$ for all $i,j$ in the extended sequence. We then define
\[
  Z_t^n(x) \;=\; \max_{1\le i\le n} Y_t^{(i)}(x), \qquad
  Z_t(x) \;=\; \lim_{n\to\infty} \frac{Z_t^n(x)-b_n(x)}{a_n(x)},
\]
provided this limit exists in distribution (Theorem~\ref{thm:gev}). 
Given a dataset $S = \{(x^{(i)},t^{(i)},z^{(i)})\}_{i=1}^{N}$, sampled from a $(X,T,Z)$ where $Z$ is the limiting tail distribution for $Y$. Our goal is to estimate the effect of the treatment assignment on the tail indices of the potential outcomes $Y_0$ and $Y_1$ as well as the conditioned potential outcomes $Y_0|X=x$ and $Y_1|X=x$.
This is challenging since the potential outcomes are only partially observed and are subject to \textit{selection bias}. 
To address this, we investigate the identifiability of the conditional maximum likelihood estimator under tail unconfoundedness. 
When presenting the method, we assume that the domain of attraction of each individual outcomes can be estimated, and we assume repeated observations for each individual. 
In practice, only one factual outcome per individual is observed, making it difficult to estimate the maxima of outcomes through standard preprocessing methods (such as the block maxima method). 
We provide a practical way of dealing with this in Appendix~\ref{sec:max_sampler}.
Let \( x \in \mathcal{X} \) and consider the potential outcomes \( Y_1 \) and \( Y_0 \). The outcomes conditioned on the individual \( X=x \) are denoted as \( Y_1(x) = (Y_1 \mid X=x) \) and \( Y_0(x) = (Y_0 \mid X=x) \). We assume both the potential outcomes and their conditionals converge to non-degenerate distributions within the domain of attraction of the Generalized Extreme Value distributions: \( \text{GEV}(\barbelow\mu_1,\barbelow\sigma_1,\barbelow\xi_1) \), \( \text{GEV}(\barbelow\mu_0,\barbelow\sigma_0,\barbelow\xi_0) \), \( \text{GEV}(\mu_1(x),\sigma_1(x),\xi_1(x)) \), and \( \text{GEV}(\mu_0(x),\sigma_0(x),\xi_0(x)) \). In other words, it is assumed that the conditions of the Theorem of Fisher-Tippet-Gnedenko are satisfied.

\begin{remark}[Scope]
\label{rem:scope}
Our framework requires each potential outcome to lie in the max-domain of attraction of a non-degenerate GEV distribution. This holds for a broad class of continuous outcomes (e.g., biomarker levels, costs,
losses), but excludes integer-valued laws whose tails are not long-tailed in the sense $\bar F(k+1)/\bar F(k)\to 1$, in particular
the Poisson, geometric, and negative binomial distributions \citep{anderson1970extreme, shimura2012discretization}. Our estimands
are therefore defined for continuous (non-lattice) severity outcomes;
extending them to count data via discrete extreme-value models~\citep{hitz2024discrete}  is an interesting direction for future work.
\end{remark}

\begin{definition}[ETE and CETE]
\label{def:extreme_treatment_effect}
The extreme treatment effect is defined as:
$$
\barbelow{\tau}_{ext} = \barbelow{\xi}_1 - \barbelow{\xi}_0.
$$
The conditional extreme treatment effect is defined as:
$$
\tau_{ext}(x) = \xi_1(x) - \xi_0(x).
$$
\end{definition}

\paragraph{Interpretation of ETE and CETE and use cases.} ETE and CETE are functions of tail indices: the asymptotic rate at which the potential outcome survival function decays. A positive ETE implies that treated outcomes
eventually \textit{exceed sufficiently high thresholds more often}. Hence, it provides an assessment of how the treatment assignment affects the heavy tails of the potential outcomes.
ETE and CETE neither identify survival functions nor a finite threshold risk difference $P(Y_1>c)-P(Y_0>c)$. Those quantities also depend on location and scale and the threshold $c$ itself. Natural applications are continuous severity biomarkers, financial losses, and costs, where ETE
screens for whether an intervention makes the tail heavier. CETE supports heterogeneous-risk analysis: it flags covariate regions in which treatment makes the tail heavier. These uses remain within the continuous, non-lattice GEV-domain scope of
Remark~\ref{rem:scope}. In practice, tail unconfoundedness cannot be tested for and in this light \textit{Estimates should not drive
safety-critical decisions without application of domain knowledge}. We introduced two metrics: $\varepsilon_{\text{ETE}}$ and  $\varepsilon_{\text{CETE}}$ which respectively quantify the errors in estimating ETE, and CETE.

\begin{definition}
\label{def:performance}
Given estimators $\hat {\barbelow \tau}_{ext}$ for ETE and $\hat{\tau}_{ext}(x)$ for CETE, the performance errors are respectively defined as:
$$
\varepsilon_{\text{ETE}} = |\hat{\barbelow\tau}_{ext} - \barbelow \tau_{ext}|
$$
and
$$
\varepsilon_{\text{CETE}}  = \mathbb{E}\left[\left(\hat{\tau}_{ext}(x) - \tau_{ext}(x)\right)^{2}\right].    
$$
\end{definition}

For our proposed methods, we will use this precision measure for quantifying the performance of different methods in our empirical studies. 

\subsection{Mathematical Motivation}
In this section, we provide two mathematical interpretations of the ETE and CETE to motivate why these quantities are interesting in a causal inference setting.

\paragraph{Regularly Varying Perspective.}
We can relate the ETE and CETE to the different rates of tail decays of a function through regular variation.
We review some definitions on regularly varying functions from~\citet[Chapter 1]{resnick2007heavy} and refer the reader to the chapter for additional information.
A function $U: \mathbb{R}_{+} \rightarrow \mathbb{R}_{+}$ is regularly varying ($RV$) if $\exists \, \xi \in \mathbb{R}$ such that
$$
\forall x>0, \lim _{t \rightarrow \infty} \frac{U(t x)}{U(t)}=x^\xi .
$$
This definition is due to Karamata~\citep{bingham1989regular}.
The parameter $\xi$ is called the exponent of variation, and we write $U \in R V(\xi)$ to mean that $U$ is a regularly varying function with exponent of variation $\xi$. 
We will present the results for the Fr\'echet case where $\xi>0$. 
The following remark characterizes our motivation~\cite{resnick1988association}
\begin{remark}
\label{remark:regular_varying}
    Let $\barbelow{\xi}_0, \barbelow{\xi}_1 > 0$. 
    Let $S_1(Y_1), S_0(Y_0)$ be the survival distributions of the potential outcome variables with positive support for the treatment and control groups respectively.
    Then if $\barbelow{\tau}_{ext} > 0$, then ${S_0(Y_0)}/{S_1(Y_1)} \in RV(\tfrac{1}{\barbelow\xi_1}-\tfrac{1}{\barbelow\xi_0})$.
    Moreover, $S_0(Y_0) < S_1(Y_1)$ for $Y_t \gg 1$.
\end{remark}
The proof for this remark relies on the Karamata representation theorem~\citep{galambos1973regularly}, with further details in Appendix~\ref{sec:proofs}. 
Therefore, the ETE and CETE can be interpreted as characterizing the ratio of the survival functions between the treatment and control groups.

\paragraph{Point Process Perspective.}
Another practical viewpoint of EVT involves interpreting the GEV in terms of the intensity of a Poisson point process~\citep[Chapter 7]{coles2001introduction}.
The intensity of the Poisson process is given by the negative log of the density in~\eqref{eq:gev}.
One can use this form to calculate the expected number of extreme points within a region of the state space.
We formalize this in the following remark:
\begin{remark}[Expected Number of Events]
\label{remark:point_process}
    Consider a causal inference task with potential outcomes variable $Y_t$, $t\in\{0,1\}$, that $Y_t$ is in the maximum domain of attraction of a GEV with rate $\barbelow{\xi}_t$, and ETE given by $\barbelow{\tau}_{ext} > 0$.
    Suppose that $\mathcal{R} \subset \mathbb{R}$ is the Borel set of interest that should be quantified for the causal inference task with $|\inf \mathcal{R}| \gg 0$.
    Let $a_n>0, b_n\in\mathbb{R}$ be the normalizing sequences from Theorem~\ref{thm:gev}, and define the normalized counting process
    
    \[
      N_{n,t}(\mathcal{R}) \;=\; \sum_{i=1}^{n} \delta_{\left(\frac{Y_t^{(i)}-b_n}{a_n}\right) \in \mathcal{R}}, \qquad t\in\{0,1\}.
    \]
    
    Then, as $n\to\infty$,
    $\mathbb{E}[N_{n,Y_1}(\mathcal{R})] > \mathbb{E}[N_{n,Y_0}(\mathcal{R})]$, where $N_{n,t}(\mathcal{R})$ is a counting process over the normalized outcomes on some Borel subset of $\mathbb{R}$.
\end{remark}

The analog of this remark can easily be made for the CETE case by conditioning on the individual covariates. 
This is particularly important, for example, in a medical setting where one may want to consider the number of severe adverse events in a treatment group versus a control group.
If the outcome variable being above a certain threshold corresponds to a serious adverse event, then the ETE and CETE can be used to quantify the number of events. 

\subsection{Tail unconfoundedness and identifiablity}
In order for the problem of estimating ETE and CETE to be tractable, we need to establish identifiability of the statistics involving counterfactual quantities in Definition~\ref{def:ident} before these quantities can be estimated from observational quantities.  
To this end, in addition to Assumptions~\ref{ass:consistency} and~\ref{ass:pos}, we will introduce a new assumption in order to guarantee identifiability. 
\begin{assumption}[Tail Unconfoundedness]
\label{ass:tail_unconf}
$$
\forall t \in \mathcal{T}, \; {P(Z_{t}\mid X,T) = P(Z_{t}\mid X)}
$$
where $Z_t$ is the MDA of the potential outcome $Y_t$.
\end{assumption}
Note that Assumption~\ref{ass:tail_unconf} is a weaker condition than the unconfoundedness presented in Assumption~\ref{ass:unconf}, which is assumed in ATE and CATE estimation~\citep{imbens2004}. 
Assumption~\ref{ass:unconf} requires that all the confounding variables are observed.
On the other hand, Assumption~\ref{ass:tail_unconf} only requires observing the confounding variables that affect the tail. This is an important relaxation since, even if some covariates that do not affect the tail are missing, ETE and CETE can still be estimated. 

\textbf{Example where tail unconfoundedness is strictly weaker than
unconfoundedness.} To illustrate how the relaxation in the assumption is
important, we present a continuous example of a
distribution family that satisfies Assumption~\ref{ass:tail_unconf} but not
Assumption~\ref{ass:unconf}. Let $P(Y_1\mid X=x, T=t)$ follow a Beta distribution
$\mathcal B(\alpha,\beta)$ with parameters $\alpha=t+1$ and $\beta=x$ for
$x,t>0$. The upper-tail index of $\mathcal B(\alpha,\beta)$ is
$\xi=-1/\beta$, so here $\xi(x)=-1/x$ depends only on $x$, not on $t$,
whereas the full conditional density depends on $t$ through $\alpha=t+1$.
Hence Assumption~\ref{ass:tail_unconf} holds while Assumption~\ref{ass:unconf} fails. Additional details are provided in Appendix~\ref{sec:tail_examples}.
Moreover, Assumption~\ref{ass:unconf} implies Assumption~\ref{ass:tail_unconf} formalized in the proposition below.
\begin{proposition}
    \label{prop: tail_unconf}
    If unconfoundedness (Assumption~\ref{ass:unconf}) holds then so does the tail unconfoundedness (Assumption~\ref{ass:tail_unconf}). 
\end{proposition}

It is important to note that tail unconfoundedness is not testable from observational data in practice, and latent
drivers of catastrophic outcomes (e.g., genetics, frailty) may be tail-relevant and correlated with treatment.
ETE/CETE estimates should therefore not drive safety-critical decisions without domain validation of the identifying assumptions. 

\subsection{Estimation of the CETE and ETE}

\textbf{Scope of the analysis.} Throughout the main analysis (from Sections~3.4 to 5), the conditional
location $\mu_t(x)$ and scale $\sigma_t(x)$ are treated as known, and only the shape $\xi_t(x)$ is
estimated. Appendix~\ref{app:joint_mu_xi} demonstrates joint estimation of $\mu(x)$ and $\xi(x)$
with $\sigma$ fixed. Joint estimation of all three parameters is not addressed.


Next, we define our conditional maximum likelihood estimator for estimating the conditional tail indices. We assume that the location and the scale parameters are known.
Let $\mathcal{H}: \mathcal{X}\times \mathcal{T}\to \mathbb{R}$ be a class of hypothesis functions estimating the shape parameter $\xi_t(x)$ for $x,t \in \mathcal{X}\times \mathcal{T}$.
In our empirical results, $\mathcal{H}$ is considered to be a class of neural networks with parameters $\theta \in \Theta$, i.e $\mathcal{H} = \{g_{\theta}: \mathcal{X}\times \mathcal{T}\to \mathbb{R}|\theta\in \Theta\}$.

\begin{definition}[CETE Estimator]
\label{def:estimator}
For $t\in\{0,1\}$, let $Z_t(x)$ be the conditional random variable $Z_t|X=x$ for the maximum domain of attraction of $Y_t|X=x$. The shape parameter $\xi_t(x)$ estimator is defined as:
$$
    \theta^* = \arg \max_{\theta \in \Theta} \sum_{j=1}^{N} \log p\left(Z_t(x^{(j)})=z_t^{(j)} \; \big | \; g_\theta\left(x^{(j)},t\right)\right)
$$
and for every $x \in \mathcal{X}$,
$$
\hat{\xi}_t(x) = g_{\theta^*}(x,t)
$$
Accordingly, we define our CETE estimator as:
$
\hat{\tau}_{ext}(x) = \hat{\xi}_1(x) - \hat{\xi}_0(x).
$
\end{definition}
Since $Z_1|X=x$ is given as the tail distribution, $p$ is taken to be the density of the conditional GEV in Definition~\ref{def:gev}. 
A naive estimator for the ETE can be defined as follows:
$$
\Tilde{\barbelow\tau} = \Tilde{\barbelow\xi}_1 - \Tilde{\barbelow\xi}_0.
$$
where,
$$
\Tilde{\barbelow{\xi}}_t = \arg \max_{\phi} \sum_{j=1}^N \log p\left(Z =z^{(j)}|T=t;\phi\right)
$$
Therefore, the naive estimator takes the samples from the GEV of the outcomes conditioned on the treatment and fits a GEV to it. This results in a biased estimator of the ETE because of the selection bias introduced by the treatment selection.
Now we define our proposed ETE estimator based on the the CETE estimator.
\begin{definition}[ETE Estimator]
\label{def:estimator_ate}
For $t\in\{0,1\}$, let $\{\hat{\xi}_t(x^{(j)})\}$ be the estimated conditional shape parameter using the CETE estimator. We sample $\{\hat{z}^{(j)}_t\}_{j=1}^N$ from the GEV with the estimated shape parameters.
The tail index $\barbelow\xi_t$ estimator is defined as:
$$
\widehat{\barbelow{\xi}}_t = \arg \max_{\phi} \sum_{j=1}^N\log p\left(\hat{Z}_t =\hat{z}^{(j)}_t;\phi\right)
$$

Accordingly, we define our ETE estimator as:
$
\hat{\barbelow\tau} = \hat{\barbelow\xi}_1 - \hat{\barbelow\xi}_0.
$
\end{definition}
The following theorem makes the defined causal quantities statistically tractable. 
\begin{proposition}[Identifiability]
\label{prop:identifiability}
Under Assumptions~\ref{ass:consistency}, \ref{ass:pos}, and~\ref{ass:tail_unconf}, the CETE and ETE estimators defined in Definition~\ref{def:estimator} and Definition~\ref{def:estimator_ate} respectively are causally identifiable. 
\end{proposition}
The identifiability theorem reduces our original causal statistics (that depend on tails of the potential outcomes) to quantities that only depend on observable variables. 
However, even under an identifiability assumption, the estimator in Definition~\ref{def:estimator} may not be easy to compute due to the fact that only one observation per individual may exist. 
We include both theoretical discussions  and empirical results for the case that only one outcome is observed per individual in Appendix~\ref{sec:max_sampler}.

\begin{remark}
\label{remark:consistency} 
Under Assumptions~\ref{ass:consistency}, \ref{ass:pos}, \ref{ass:tail_unconf}, and~\ref{a1}--\ref{a6}
(Appendix~\ref{sec:max_sampler}), with known location and scale, the CETE estimator is consistent:
$\hat\tau_{ext}(x)\xrightarrow{p}\tau_{ext}(x)$ for $P_X$-a.e.\ $x$ (Theorem~\ref{thm-cete}).
\end{remark}
Consistency and asymptotic normality of the full procedure of the proposed ETE estimator are not established here; we assess it empirically in Section~5.

\section{Computational Method}
\label{sec:method}

In this section, we present the computational method employed in our study. 

\paragraph{End-to-end pipeline.} Estimation proceeds in three stages.
(1)~\emph{Surrogate maxima.} Each individual has only one factual outcome. The $\varepsilon$-max-sampler
therefore clusters the units of each treatment arm into $K$ groups by $k$-means and keeps the
largest outcome in each cluster. This treats the $m=N/K$ outcomes of nearby units as approximate
replicates of a single individual, giving surrogate draws of $Z_t(x)$.
(2)~\emph{Conditional tail indices.} For each $t$, we estimate $\xi_t(x)$ by maximizing the conditional
GEV likelihood (Definition~\ref{def:estimator}) on the arm-$t$ surrogate maxima. This gives
$\hat\tau_{ext}(x)=\hat\xi_1(x)-\hat\xi_0(x)$.
(3)~\emph{Marginalization.} For every unit $x^{(j)}$ in the full sample (both arms), we draw
$\hat z_t^{(j)}\sim\mathrm{GEV}(\mu_t(x^{(j)}),\sigma_t(x^{(j)}),\hat\xi_t(x^{(j)}))$. These draws
approximate the law of $Z_t$ under $P_X$ rather than under $P_{X\mid T=t}$, which removes the
selection bias of the naive estimator. Fitting one GEV to $\{\hat z_t^{(j)}\}$ gives
$\hat{\barbelow\xi}_t$ and hence $\hat{\barbelow\tau}_{ext}$.
The theory maps onto some of the stages as follows. \emph{Identification}
(Proposition~\ref{prop:identifiability}) is a population statement about the limit $Z_t$; it
justifies fitting arm-$t$ data in stage~2. The \emph{sampler approximation} (Lemma~\ref{lem1})
bounds the gap between the stage-1 surrogates and true replicate maxima by $L_F\varepsilon$.
\emph{Consistency} (Theorems~\ref{thm1}--\ref{thm-cete}) combines the two for stage~2. 

Given a dataset $\{(x_i,t_i,y_i)\}_{i=1}^N$, we first apply a variation of the block maxima method. We then learn the different parameters of the limiting GEV by maximizing the conditional log-likelihood of the maximum of the outcomes to the covariates. The CETE and ETE can then be computed using the estimated quantities on $\xi$, a pseudocode of the algorithm is provided in Appendix~\ref{sec:algo}.

\textbf{Empirical $\varepsilon$-max-sampler.}
As noted, the standard procedure of using the block maxima in time cannot be used for practical scenarios of interest in the causal setting since we do not have a temporal component --- we only observe one outcome for every individual.
To approach the task of approximating the maximum of an outcome variable across multiple copies, we instead consider the analogy of the block maxima in \emph{space} to exploit the property that individuals with similar characteristics tend to have similar outcomes.
As a result, we can use realizations of individuals locally to approximate independent and identical realizations of a single individual. 
As the number of points increases and the ball size around every individual converges to zero the approximation of multiple copies of individual outcomes becomes increasingly accurate when taking the maximum of their outcomes.

To do this, we use the $k$-means algorithm to identify clusters of individuals with similar feature profiles. 
Then, within each cluster, we select the individual with the largest outcome as an approximation of the maximum outcome for that group of similar individuals. We provide a detailed description of the $\varepsilon$-max-sampler procedure in Appendix~\ref{sec:algo} as well as a theoretical analysis of this empirical approach in Appendix~\ref{sec:max_sampler}.
\paragraph{Likelihood-Based Estimator.}
Since we assume that the potential outcomes are in the maximum domain of attraction of a GEV, we maximize the likelihood of a GEV parameterized by functions mapping the observations to the parameters. 
Let $\theta_t = (\mu_t,\sigma_t,\xi_t)$ 
where $\mu_t : \mathcal{X} \to \mathbb{R}$, $\sigma_t : \mathcal{X} \to \mathbb{R}_+$ and $\xi_t : \mathcal{X} \to \mathbb{R}$ for $t\in \{0,1\}$. 
Let $\{(x^{(m_i)},t^{(m_i)},y^{(m_i)})\}_{i=1}^{K}$ $K$ be the data points after computing the $\varepsilon$-max-sampler. Then we have for $t \in \{0,1\}$,
$$
\begin{aligned}
&\hat{\theta}_t = \arg \max_{\theta} \sum_{i=1}^{K} \log(\ell(\theta \mid x^{(m_i)},t^{(m_i)},y^{(m_i)})) \mathbbm{1}_{(t^{m_i} = t)}
\end{aligned}
$$
where
$$
\ell(\theta \mid x,t,y) = \frac{\Phi(y\mid x,t)^{\xi_t(x)+1}}{\sigma_t(x)}  e^{-\Phi(y \mid x,t)} \; $$
and
$$
\Phi(y\mid x,t)= \left(1+\xi_t(x)\left(\frac{y-\mu_t(x)}{\sigma_t(x)}\right)\right)^{-1 / \xi_t(x)}.
$$

We write the likelihood for the full triple for generality, in all main-text experiments $\mu_t$ and
$\sigma_t$ are fixed at their known values and only $\xi_t$ is learned.

After fitting the parameters, we can estimate the CETE as the following
$
\hat{\tau}_{ext}(x) = \hat{\xi}_1(x) - \hat{\xi}_0(x)
$. 
In our synthetic experiments, we note that fitting all three conditional parameters of the Generalized Extreme Value (GEV) distribution can be numerically unstable. 
Given that our primary objective is to investigate the variability of the shape parameter, we assume that the location and scale parameters are known. 
Our focus is to estimate the shape parameter as a function of the covariates, in order to accurately estimate the CETE and the ETE. 
It is crucial to highlight that the shape parameter is our parameter of interest since it characterizes the heaviness of the tail—a feature that is not encapsulated by pointwise information. 
This characteristic cannot be adequately described by the moments of the distribution or by its location and scale parameters alone.


\begin{figure}
    \centering
    \includegraphics[width=\textwidth]{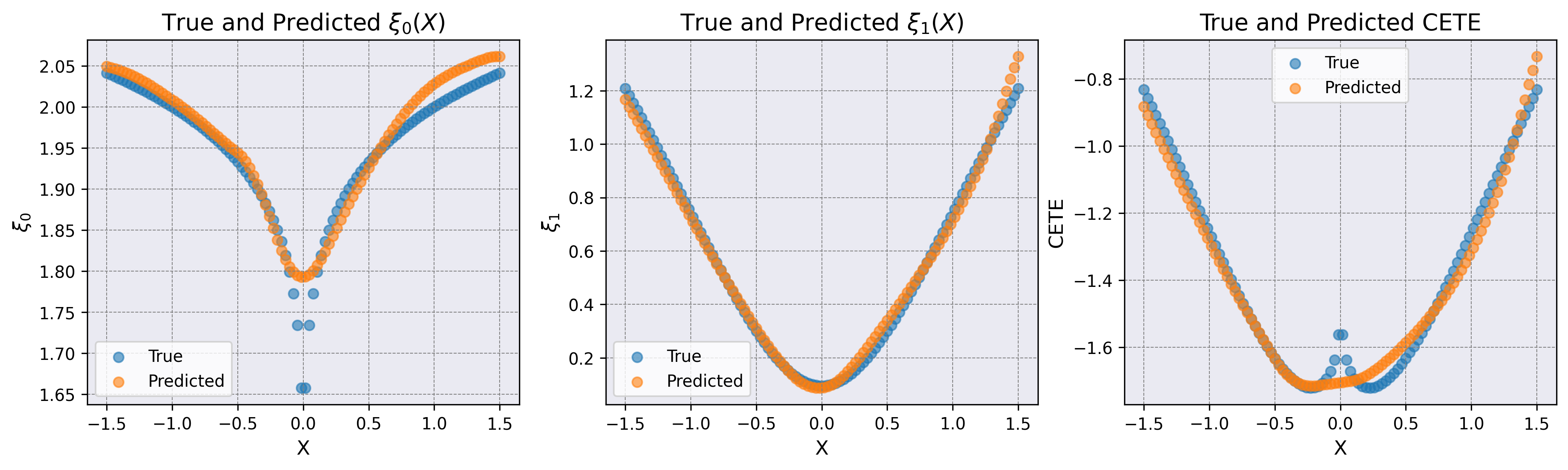}
    \caption{From left to right, the figures represent the estimation of the individual tail parameter under the control condition, the tail parameter under the treatment condition, and the Conditional Extreme Treatment Effect (CETE) using our proposed method.}
    \label{fig:exp_synthetic}
\end{figure}

\section{Experiments}
%
\begin{figure}[t!]
    \centering
    \includegraphics[width=0.45\textwidth]{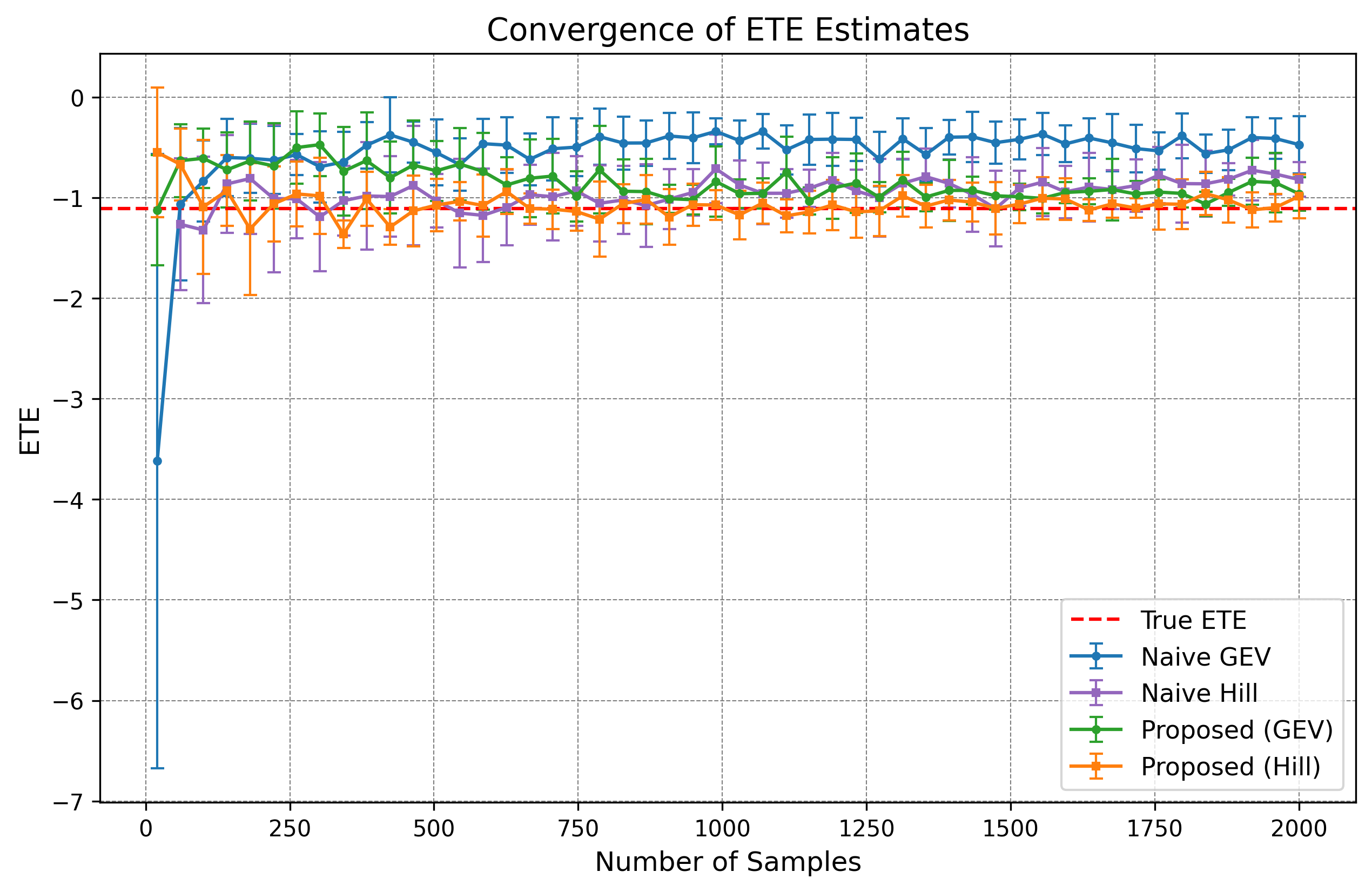}
    \caption{Convergence of four ETE estimators with sample size on the synthetic 1D setting: Naive GEV, Naive Hill, Proposed (GEV), and Proposed (Hill). Error bars show variability across simulations; the dashed red line marks the true ETE.}
    \label{fig:ete_convergence}
\end{figure}

In this section, we demonstrate the empirical efficacy of the proposed framework. First, to verify the estimation procedure for conditional GEV distributions, we evaluate the proposed algorithm in a synthetic 1D scenario. 
Subsequently, we apply our approach to several semi-synthetic high-dimensional datasets.

\textbf{Baseline Treatment Effect Models.} The proposed framework can be used in conjunction with existing CATE estimation methods. As such, we consider a few models used in the literature and adapt them to the proposed framework for estimating the CETE. The {\bf $s$-learner} method treats the treatment variable as another covariate and builds a neural network $f(x,t)$ to estimate the potential outcomes and then estimates the ITE as $\hat{\tau}(x) = f(x,1) - f(x,0)$.  
The {\bf $t$-learner} builds two neural networks, one for each treatment group $f_1(x)$ and $f_0(x)$, and the individual outcome is given by $\hat{\tau}(x) = f_1(x) - f_0(x)$ ~\citep{kunzel2019metalearners}.
The {\bf TARNet} model is a framework for estimating conditional average treatment effects with counterfactual balancing~\citep{shalit}.
It consists of a pair of functions, $(\Phi, h)$, where $\Phi$ is a representation learning function and $h$ is a function that learns the two potential outcomes in the representation space. 
CATE is then estimated as $\hat{\tau}(x) = h(\Phi(x),1) - h(\Phi(x),0)$ 
The TARNet model minimizes the objective $\mathcal{L}(\Phi,h)$, which balances the model's performance on factual data and the similarity of the representations in the latent space using the integral probability metric (IPM).
The balancing weight, $\alpha$, controls the trade-off between the similarity and model performance. When $\alpha =0$ the model is denoted by TARNet, when $\alpha>0$ the model is known as counterfactual regression (CFR). It corresponds to the 1-Wasserstein distance for 1-Lipschitz functions \citet{villani}.

\subsection{Synthetic Data Experiments}

Next we validate our proposed method by an experiment using synthetic data. We generate $5000$ samples for $Y_0|X=x$ from a Generalized Extreme Value (GEV) distribution, $\text{GEV}(\exp(x), 1, \xi_0(x))$, where $\xi_0(x) = 1 + |x|^{0.1}$, and for $Y_1|X=x$ from $\text{GEV}(x^2, 1, \xi_1(x))$, with $\xi_1(x) = \log(1.1 + x^2)$. The feature variable $X$ is sampled from a normal distribution, $\mathcal{N}(0,1)$, and the propensity score is set to $0.3$ for $X>0$ and $0.7$ otherwise. We evaluate the model on data sampled from $\mathcal{N}(0,1)$ and with a propensity score $p(t=1|x) = \sigma(50x^2 - 5)$, where $\sigma$ denotes the sigmoid function. 

The results for this experiment are presented in Figure~\ref{fig:exp_synthetic}. 
We observe that the model accurately learns the shape parameter as a function of the covariates, which leads to precise estimation of the Conditional Extreme Treatment Effect (CETE). 
We then use this trained neural network to estimate the Extreme Treatment Effect (ETE) by sampling from the learned GEV distribution for potential outcomes, since we have learned the shape parameter. 
We fit a GEV to the marginal samples of the outcomes $\{\hat{z}^{(j)}_t\}_{j=1}^N$. 
We also evaluate the naive estimator by fitting a GEV to the observed treatment outcomes $\{z^{(j)}_1\}_{j=1}^{N_1}$ and $\{z^{(j)}_0\}_{j=1}^{N_0}$, where $N_1$ and $N_0$ are the cardinalities of the treatment and control groups, respectively. 
We observe that our estimator exhibits unbiased behavior, unlike the naive estimator. 
The convergence behavior of both the naive and the proposed ETE estimators are presented in Figure~\ref{fig:ete_convergence}. To assess our parametric EVT approach against a non-parametric alternative, we compare four estimators of the ETE on the synthetic 1D setting: the naive block-maxima GEV estimator, the naive Hill estimator applied directly to the
treatment/control samples, our proposed marginalization procedure with a GEV likelihood, and the same marginalization procedure with the Hill estimator substituted for the marginal-fit step. This comparison is included in Figure~\ref{fig:ete_convergence}.

\begin{table*}[t]
\centering
\caption{Comparison of the CETE performance of our proposed estimator. The error bars are calculated for various random seeds.}
\label{tab:results_cete}
\begin{tabular}{@{}lccc@{}}
\toprule
\textbf{Model} & \textbf{IHDP (original)} & \textbf{IHDP (Fr\'echet)} & \textbf{IHDP (Weibull)} \\
\midrule
S-Learner & $0.048 \pm 0.028$ & $0.396 \pm 0.060$ & $0.307 \pm 0.070$ \\
T-Learner & $0.188 \pm 0.053$ & $0.634 \pm 0.307$ & $0.345 \pm 0.004$ \\
TARNet & $0.119 \pm 0.071$ & $0.223 \pm 0.105$ & $0.278 \pm 0.172$ \\
CFR(Wass) & $0.069 \pm 0.004$ & $0.211 \pm 0.104$ & $0.261 \pm 0.061$ \\
\bottomrule
\end{tabular}
\end{table*}

\subsection{Semi-synthetic data experiments}

Having validated the proposed method in estimating conditional GEVs, we next consider application to computing the ETE and the CETE.

\textbf{Dataset descriptions.} We present a representative family of semi-synthetic datasets tailored for estimating the CETE. 
Specifically, we use the Infant Health and Development Program (IHDP) dataset~\citep{ramey1992infant} and two variants corresponding to different MDAs. 
Both datasets consider potential outcomes generated based on known functions of the observed covariates. 
The selection bias is the same in the three datasets, provided by the real-world experiment, and we produce the potential outcomes following different GEVs. 
We provide more detailed descriptions of these datasets in Appendix~\ref{sec:description}.

\paragraph{CETE results.} We present the performance for CETE estimation in Table \ref{tab:results_cete}. 
The results suggest that the CFR(Wass) model performs the best for the Fr\'echet and the Weibull cases while the S-Learner performs best for the IHDP original case (the Gumbel case).

\paragraph{ETE results.} 
We consider the performance of our proposed ETE estimator compared to that of a naive estimator, which is based on the block maxima method applied directly to the potential outcomes. 
The naive estimator is then given by: $\barbelow\tau_{\text{naive}} = \barbelow \xi_{1,\text{naive}} - \barbelow \xi_{0,\text{naive}}$, where $\xi_{t,\text{naive}}$ correspond to the estimated shape parameter for the MDA of the conditional outcome $Y|T=t$. It is important to note that the conditional outcome $Y|T=t$ is distinct from the potential outcome $Y_t$. The ETE results are summarized in Table~\ref{tab:results_ete}
\renewcommand*{\arraystretch}{0.5}
\begin{table}[H]
\centering
\caption{ETE Performance comparison under the naive and proposed estimators across different datasets using the CFR(Wass) model.}
\label{tab:results_ete}
\footnotesize
\begin{sc}
\begin{tabular}{lcc}
\bf{Dataset} & {\bf $\varepsilon_{\text{ETE}}$(Naive)} & {\bf $\varepsilon_{\text{ETE}}$(Ours) }\\ 
\bottomrule \\
IHDP (original) & $0.094$ & $0.001 $  \\
\midrule
IHDP (Fr\'echet) & $0.346$ & $0.095$  \\
\midrule
IHDP (Weibull) & $0.931$ & $0.827$  \\
\bottomrule
\end{tabular}
\end{sc}
\end{table}


\section{Conclusions}
In this work, we proposed a new framework for estimating extreme treatment effects. We defined new statistical quantities CETE and ETE capturing extreme treatment effects. 
We studied the identifiability and the consistency of the proposed estimators.
To circumvent the lack of tail data, we proposed a practical approach to sample from the MDAs and introduce a neural network approach for learning the GEV parameters. Our contributions offer new methods of estimation of extreme treatment effects. 

\paragraph{Limitations}
The main limitation of the method is we require regularity on the space of conditional outcomes to be able to use our block maxima surrogate. Circumventing this requires multiple outcomes per individual, which is not practical in most cases. 
 We also note that tail unconfoundedness cannot be fully verified in practice. In this light,  the developed metrics should not drive safety measures without incorporating domain knowledge. Furthermore, our theoretical guarantees only apply to the CETE estimator. Establishing consistency and asymptotic normality of the ETE estimator is an interesting direction for future work.

\subsubsection*{Acknowledgments}
This work was supported in part by the National Science Foundation (NSF) under the National AI Institute for Edge Computing Leveraging Next Generation Wireless Networks Grant \#2112562. Ali Hasan was supported in part by the Air Force Office of Scientific Research under award number FA9550-20-1-0397.

\bibliography{tmlr}
\bibliographystyle{tmlr}

\appendix
\newpage

\appendix

{\Huge \textbf{Appendix}} \\

\section{Simulating GEV Potential Outcomes}
\label{sec:motivation}
In this illustration, we model potential outcomes under control and treatment scenarios for a population of $n=5$ individuals. Each individual $k \in \{1,\ldots,n\}$ is associated with two potential outcomes. 
For the control scenario, the potential outcome for individual $k$ follows a Generalized Extreme Value (GEV) distribution with a location parameter $\mu^0_k=5+k$, a scale parameter $\sigma^0_k=1$, and a shape parameter $\xi^0_k =0.5$. 
For the treatment scenario, the potential outcome for individual $k$ follows a GEV distribution with a location parameter $\mu^1_k=6+k$, a scale parameter $\sigma^1_k=1.5$, and a shape parameter $\xi^1_k =1.5+0.1k$.

Figure~\ref{fig:individual_extrmes} illustrates the distribution of the potential outcomes for the $n=5$ individuals, each with distinct location and shape parameters.
In our work, we study the differences in the $\xi$ parameter, which, as the figure shows, dictates how the tail of the distribution behaves. 
\begin{figure}[h]
    \centering  
    \includegraphics[width=5in]{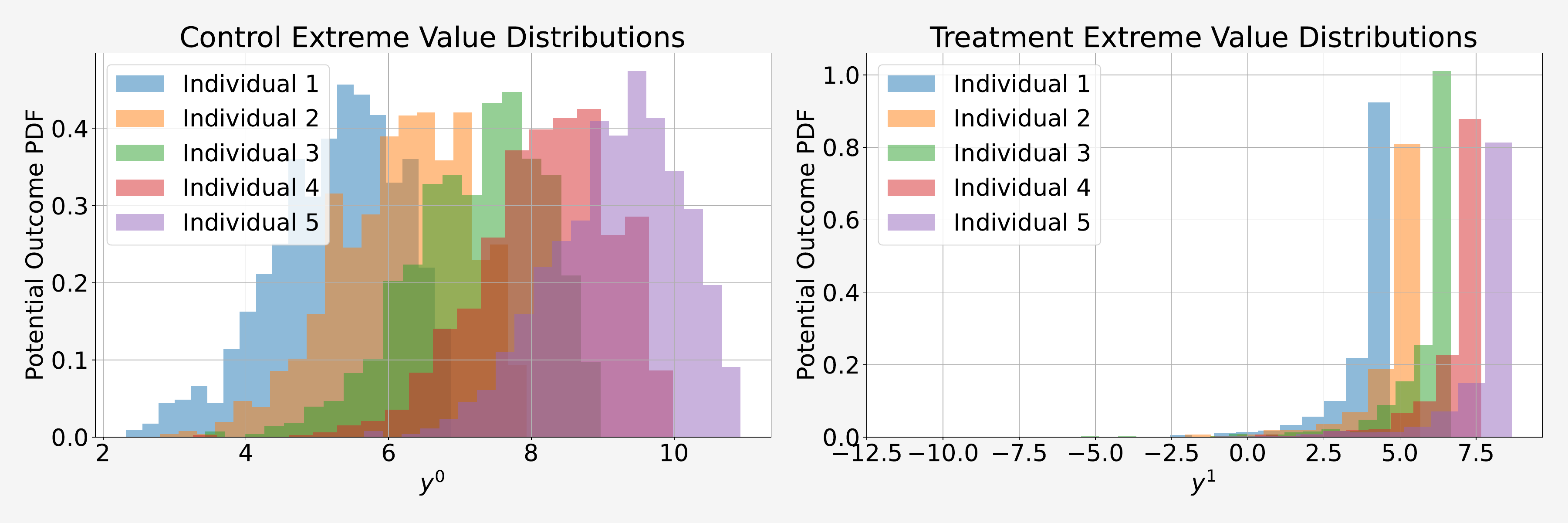}
    \caption{
 Different conditional outcomes for different treatment groups. The goal of our method is to estimate the difference in the shape parameters of these limiting GEVs.}
    \label{fig:individual_extrmes}
\end{figure}

 \section{Theoretical Results}
\label{sec:proofs}
In this section, we begin by providing proofs for the remarks presented in the paper, which support the proposed ETE and CETE definitions. This is followed by a discussion of other potential definitions for the extreme treatment effects. Subsequently, we present proofs for the causal inference results presented in the main text. Finally, we provide a theoretical analysis of the  $\varepsilon$-max sampler algorithm utilized in our empirical settings.
\subsection{Proof of Regularly Varying Remark}
\label{sec:proof_rv}

\begin{remark1}
Let $\xi_0,\xi_1>0$. Let $S_1(Y_1),S_0(Y_0)$ be the survival distributions
of the potential outcome variables with positive support for the treatment
and control groups respectively. Then if $\tau^{ext}>0$,
$S_0(Y_0)/S_1(Y_1)\in RV\big(\tfrac{1}{\xi_1}-\tfrac{1}{\xi_0}\big)$.
Moreover, $S_0(Y_0)<S_1(Y_1)$ for $Y_t\gg 1$.
\end{remark1}

\begin{proof}
Define $F_0,F_1$ as the CDFs of the outcome variables for the control and
treatment groups respectively. Then $S_0=1-F_0$ and $S_1=1-F_1$.
Under the GEV parameterization of Eq.~1, a distribution in
the Fr\'echet maximum domain of attraction with shape parameter $\xi_t>0$
has a regularly varying survival function with index $-1/\xi_t$. By the Karamata characterization theorem,
we can write
\[
  S_0(y) = y^{-1/\xi_0}L_0(y),
  \qquad
  S_1(y) = y^{-1/\xi_1}L_1(y),
\]
where $L_0,L_1$ are slowly varying. As $y\to\infty$, the polynomial
component dominates, and
\[
  \lim_{y\to\infty}\frac{S_0(y)}{S_1(y)}
  \;\sim\; \frac{y^{-1/\xi_0}}{y^{-1/\xi_1}}
  \;=\; y^{1/\xi_1-1/\xi_0}.
\]
Since $\tau^{ext}=\xi_1-\xi_0>0$ implies $1/\xi_1<1/\xi_0$
(for $\xi_0,\xi_1>0$), the exponent
$1/\xi_1-1/\xi_0<0$, so $S_0(y)/S_1(y)\to 0$, i.e.\ $S_0(y)<S_1(y)$ for
$y\gg 1$. Thus, the difference in the tail indices can be seen as the
ratio of the survival functions for regions deep in the tail.
\end{proof}

\subsection{Proof of Point Process Remark}
\label{sec:proof_pp}
\begin{remark2}[Expected Number of Events]
    Consider a causal inference task with potential outcomes variable $Y_t$, $t\in\{0,1\}$, that $Y_t$ is in the maximum domain of attraction of a GEV with rate $\xi$, and ETE given by $\barbelow{\tau}_{ext} > 0$.
    Suppose that $\mathcal{R} \subset \mathbb{R}$ is the Borel set of interest that should be quantified for the causal inference task with $|\inf \mathcal{R}| \gg 0$.
    Let $a_n>0, b_n\in\mathbb{R}$ be the normalizing sequences from Theorem~\ref{thm:gev}, and define the normalized counting process
    
    \[
      N_{n,t}(\mathcal{R}) \;=\; \sum_{i=1}^{n} \delta_{\left(\frac{Y_t^{(i)}-b_n}{a_n}\right) \in \mathcal{R}}, \qquad t\in\{0,1\}.
    \]
    
    Then, as $n\to\infty$,
    $\mathbb{E}[N_{n,Y_1}(\mathcal{R})] > \mathbb{E}[N_{n,Y_0}(\mathcal{R})]$, where $N_{n,t}(\mathcal{R})$ is a counting process over the normalized outcomes on some Borel subset of $\mathbb{R}$.
\end{remark2}
\begin{proof}
From our assumptions on $\mathcal{R}$, we suppose that the GEV is the appropriate description of the events.
Additionally, suppose that both $Y_0$ and $Y_1$ have the same limiting sequence $a_n, b_n$ and by EVT 
$$
\mathbb{E}[N_{Y_0}(\mathcal{R})] = \left [ 1 + \xi_0 c_\mathcal{R} \right]^{-1/\xi_0}
$$
and 
$$
\mathbb{E}[N_{Y_1}(\mathcal{R})] = \left [ 1 + \xi_1 c_\mathcal{R} \right]^{-1/\xi_1}.
$$
Since $\xi_1 > \xi_0$,
\begin{align*}
1+\xi_1 c_\mathcal{R} &> 1 + \xi_0 c_\mathcal{R} \\
(1+\xi_1 c_\mathcal{R})^{1/\xi_1} &> (1 + \xi_0 c_\mathcal{R})^{1/\xi_1} \\ 
(1+\xi_1 c_\mathcal{R})^{1/\xi_1} &< (1 + \xi_0 c_\mathcal{R})^{1/\xi_0} \\ 
(1+\xi_1 c_\mathcal{R})^{-1/\xi_1} &> (1 + \xi_0 c_\mathcal{R})^{-1/\xi_0}
\end{align*}
with $c_\mathcal{R} \in \mathbb{R}_+ \cup \{0 \}$.
\end{proof}

\subsection{Possible Definitions for Extreme Treatment Effects}
In this section, we note that while we motivate our definition for the extreme treatment effect from an intuitive perspective, a point process perspective, and from a regular variation perspective, it is possible to define other functions of the shape parameter $\Psi(\xi_1(x), \xi_0(x))$ that capture other properties. Our presented methods can be extended to these general functions, as the main problem in causal inference lies in estimating $\xi_1(x)$ and $\xi_0(x)$. Moreover, the choice of the difference between the shape parameters guarantees that the statistical properties of the estimators of shape parameters (e.g., expectation and variance) can be easily extended to the estimator of the extreme treatment effects, which may not be the case for other functions such as the ratio of the shape parameters. It is important to highlight that, at times, quantifying the proportional change of the shape parameter can be essential (thus, taking the ratio would be more justified). In contrast, in some situations, the magnitude of change itself is significant regardless of the initial value (for instance, if having a heavier tail above $2$ is critical and the control shape parameter is $0$, taking the difference provides more information). Therefore, different definitions might be more relevant depending on the context
   
\subsection{Causal Inference Proofs}
We restate the main statements and then state the proofs.
\begin{proposition1}
    If unconfoundedness holds then tail unconfoundedness holds. 
\end{proposition1}
\begin{proof}
Assuming that unconfoundedness holds, and letting $t \in \mathcal{T}$, we have

$$
Y_t \ind T \mid X
$$
Therefore,
$$
P(Y_t\mid X,T) = P(Y_t\mid X).
$$
For $\{Y^{(k)}\}_{k=1}^{m}$ i.i.d copies of $Y$, we have that 
$$
\forall k \in \{1,\ldots,m\}, \; P(Y^{(k)}_t \mid X,T) = P(Y^{(k)}_t \mid X) 
$$
By taking the limit on both sides, we have
$$
\lim_{m\rightarrow \infty}{P(Z_{t}^{m}\mid X,T) = \lim_{m\rightarrow \infty} P(Z_{t}^{m}\mid X)}
$$
which implies tail unconfoundedness holds.
\end{proof}
\begin{proposition2}[Identifiability]
Under Assumptions~\ref{ass:consistency}, \ref{ass:pos}, and~\ref{ass:tail_unconf}, the CETE and ETE estimators defined in Definition~\ref{def:estimator} and Definition~\ref{def:estimator_ate} respectively are causally identifiable. Moreover, for $t \in \{0,1\}$, the proposed estimator satisfies 
$$
    \hat{\xi}_t = \arg \max_{\theta\in \Theta} \sum_{j=1}^{N}  \log p\left( Z^{(j)} \mid g_\theta(x^{(j)}),t^{(j)}\right) \mathbbm{1}_{t^{(j)} = t}.
$$
\end{proposition2}
\begin{proof}
By definition, we have 
$$
\begin{aligned}
\hat{\theta}_t = &\arg \max_{\theta} \lim_{m\to \infty} \sum_{j=1}^{N} \log P\left(Z^{m}_t(x^{(j)}) \; \big | \; \theta\left(x^{(j)}\right)\right) \\
=&\arg \max_{\theta}  \sum_{j=1}^{N}\lim_{m\to \infty} \log P\left(Z^{m}_t(x^{(j)}) \; \big | \; \theta\left(x^{(j)}\right)\right)
\end{aligned}
$$
By tail unconfoundedness, we have 
$$
\lim_{m \to \infty}P\left(Z^{m}_t(x^{(j)}) \; \big | \; \theta\left(x^{(j)}\right)\right) = \lim_{m \to \infty}P\left(Z^{m}_t(x^{(j)}) \; \big | \; \theta\left(x^{(j)}\right),T=t\right) 
$$
From the consistency assumption, it follows that
$$
\lim_{m \to \infty}P\left(Z^{m}_t(x^{(j)}) \; \big | \; \theta\left(x^{(j)}\right),T=t\right) = \lim_{m \to \infty}P\left(Z^{m}(x^{(j)}) \; \big | \; \theta\left(x^{(j)}\right),T=t\right) 
$$
Notice that the right-hand-side quantity does not involve any potential outcomes. 
We then have
$$\lim_{m \to \infty}P\left(Z^{m}_t(x^{(j)}) \; \big | \; \theta\left(x^{(j)}\right)\right) = \lim_{m \to \infty}P\left(Z^{m}(x^{(j)}) \; \big | \; \theta\left(x^{(j)}\right),T=t\right)
$$
Hence, the estimator is identifiable.
\end{proof}

\section{Theoretical Analysis for $\varepsilon$-max-sampler}
\label{sec:max_sampler}

\begingroup
Fix $t\in\{0,1\}$. We estimate $\xi_t(x;\theta)$, $\theta\in\Theta\subset\mathbb R^p$,
by maximum likelihood on surrogate maxima produced by the $\varepsilon$-max sampler.
Write $\ell(\theta\mid x,t,y)$ for the conditional GEV density.
\begin{definition}[Extended probability space for $Z_t(x)$]
\label{def:space}
For each $x\in\mathcal X$, let $(Y_t^{(i)}(x))_{i\ge1}$ be i.i.d.\ with law
$Y_t\mid X=x$ (hypothetical replicates of a fixed individual). Set
$Z_t^m(x)=\max_{i\le m}Y_t^{(i)}(x)$, and
\[
  Z_t(x)\;:=\;\lim_{m\to\infty}\frac{Z_t^m(x)-b_m(x)}{a_m(x)}\quad\text{in distribution,}
\]
when the limit exists.
\end{definition}

\begin{assumption}[Parameter space and correct specification]\label{a1}
$\Theta\subset\mathbb R^p$ is compact and $\exists\,\theta_t^\star\in\mathrm{int}(\Theta)$
with $\xi_t(x;\theta_t^\star)=\xi_t(x)$ for $P_X$-a.e.\ $x$.
\end{assumption}

\begin{assumption}[Regularity of the log-likelihood]\label{a2}
$\theta\mapsto\log\ell(\theta\mid x,t,y)$ is continuous on $\Theta$ for
$P_{X,Y}$-a.e.\ $(x,y)$, and
\[
  \mathbb E\Big[\sup_{\theta\in\Theta}\big|\log\ell(\theta\mid X,t,Z_t(X))\big|\Big]<\infty .
\]
\end{assumption}

\begin{assumption}[Well-separated maximum]\label{a3}
$\theta_t^\star$ is the unique maximizer of $Q(\theta):=\mathbb
E[\log\ell(\theta\mid X,t,Z_t(X))]$ on $\Theta$, and for every $\delta>0$,
$\sup_{\|\theta-\theta_t^\star\|\ge\delta}Q(\theta)<Q(\theta_t^\star)$.
\end{assumption}

\begin{assumption}[Independence and Lipschitz quantile function]\label{a4}
(i) $Y_t(x)\perp Y_t(x')$ for all $x\neq x'\in\mathcal X$.
(ii) There exists $L_F<\infty$ such that
\[
  \sup_{u\in(0,1)}\big|F_t^{-1}(u\mid x)-F_t^{-1}(u\mid x')\big|
  \;\le\; L_F\|x-x'\|,\qquad \forall\,x,x'\in\mathcal X,
\]
where $F_t(\cdot\mid x)$ is the conditional CDF of $Y_t\mid X=x$.
\end{assumption}

\begin{assumption}[Score Lipschitz in outcome]\label{a5}
There exists $L:\mathcal X\times\{0,1\}\to\mathbb R_+$ with $\mathbb
E[L(X,t)]<\infty$ such that, for all $\theta\in\Theta$ and $y,y'\in\mathbb R$,
$\big|\log\ell(\theta\mid x,t,y)-\log\ell(\theta\mid x,t,y')\big|\le L(x,t)|y-y'|$,
where $L(x,t)$ does not depend on $\theta$.
\end{assumption}

\begin{assumption}[Rate condition]\label{a6}
As $N\to\infty$: the ball radius $\varepsilon(N)\to0$, $K(N)\to\infty$, and the number of clusters $K(N)=o(N)$
(equivalently $m(N):=N/K(N)\to\infty$).
\end{assumption}

\section*{Step 1: Surrogate error via coupling}

Fix $x\in\mathcal X$. Condition on a cluster configuration $x_1,\dots,x_m$
with $\max_i\|x_i-x\|<\varepsilon$. Write
\[
  \hat Z_{t,\varepsilon}^K(x):=\max_{i\le m}Y_t^{(i)},\quad Y_t^{(i)}\sim
  F_t(\cdot\mid x_i)\ \text{independent (Assumption~\ref{a4}(i))},
\]
and $Z_t^m(x)$ as in Definition~\ref{def:space}. By independence,
\begin{equation}\label{eq:marginals}
  P\big(\hat Z_{t,\varepsilon}^K(x)\le z\big)=\prod_{i=1}^m F_t(z\mid x_i),
  \qquad
  P\big(Z_t^m(x)\le z\big)=F_t(z\mid x)^m .
\end{equation}
Let $F^{-1}(u):=\inf\{y:F(y)\ge u\}$ denote the generalized inverse.

\begin{lemma}[Quantile transform]\label{lemA}
For $U\sim\mathrm{Unif}(0,1)$ and any CDF $F$, $\mathrm{Law}(F^{-1}(U))=F$.
\end{lemma}
\begin{proof}
$F^{-1}(u)\le z\iff u\le F(z)$ by right-continuity of $F$; hence
$P(F^{-1}(U)\le z)=P(U\le F(z))=F(z)$.
\end{proof}

\begin{lemma}[Max contraction]\label{lemC}
For $a,b\in\mathbb R^m$, $|\max_i a_i-\max_i b_i|\le\max_i|a_i-b_i|$.
\end{lemma}
\begin{proof}
Without loss of generality $\max_i a_i\ge\max_i b_i$; take $i^\star\in\arg\max_i a_i$. Then
$\max_i a_i-\max_i b_i\le a_{i^\star}-b_{i^\star}\le\max_i|a_i-b_i|$.
\end{proof}

\begin{lemma}[Coupling bound]\label{lem1}
Under Assumption~\ref{a4}, there exists a coupling $\pi$ of
$\mathrm{Law}(\hat Z_{t,\varepsilon}^K(x))$ and $\mathrm{Law}(Z_t^m(x))$
under which $|V-W|\le L_F\varepsilon$ almost surely, where $(V,W)\sim\pi$.
Consequently, for every $L$-Lipschitz $g$ integrable under both laws,
\begin{equation}\label{eq:transfer}
  \Big|\mathbb E\big[g(\hat Z_{t,\varepsilon}^K(x))\big]-\mathbb E\big[g(Z_t^m(x))\big]\Big|
  \;\le\; L\cdot L_F\varepsilon .
\end{equation}
\end{lemma}
\begin{proof}
\emph{Construction.} On $(\Omega',\mathcal F',P')=([0,1]^m,\mathcal
B([0,1]^m),\mathrm{Leb}^{\otimes m})$ with coordinates
$U_1,\dots,U_m\stackrel{iid}\sim\mathrm{Unif}(0,1)$, define
\[
  V_i:=F_t^{-1}(U_i\mid x_i),\quad W_i:=F_t^{-1}(U_i\mid x),\quad
  V:=\max_i V_i,\quad W:=\max_i W_i .
\]
By Lemma~\ref{lemA} and independence of $U_i$: $V_i\sim F_t(\cdot\mid x_i)$
independent, $W_i\stackrel{iid}\sim F_t(\cdot\mid x)$. Hence
$P'(V\le z)=\prod_i F_t(z\mid x_i)$ and $P'(W\le z)=F_t(z\mid x)^m$, matching
\eqref{eq:marginals}, so $\pi:=\mathrm{Law}_{P'}(V,W)$ is a valid coupling.

For each $\omega'\in\Omega'$ and each $i$,
Assumption~\ref{a4}(ii) at $u=U_i(\omega')$ gives $|V_i-W_i|\le L_F\|x_i-x\|
\le L_F\varepsilon$. Lemma~\ref{lemC} yields $|V-W|\le L_F\varepsilon$ on
$\Omega'$.

Expectations depend only on the marginal law, so $\mathbb
E[g(\hat Z_{t,\varepsilon}^K(x))]=\mathbb E_{P'}[g(V)]$ and $\mathbb
E[g(Z_t^m(x))]=\mathbb E_{P'}[g(W)]$. Then
\[
  \Big|\mathbb E[g(\hat Z_{t,\varepsilon}^K(x))]-\mathbb E[g(Z_t^m(x))]\Big|
  =\big|\mathbb E_{P'}[g(V)-g(W)]\big|\le\mathbb E_{P'}|g(V)-g(W)|
  \le L\cdot L_F\varepsilon.\qedhere
\]
\end{proof}

\section*{Step 2: Consistency of $\hat\theta_t$}

Define, for $\theta\in\Theta$:
\[
  \hat Q_N(\theta):=\tfrac1N\sum_{j=1}^N\log\ell\big(\theta\mid
  x^{(j)},t,\hat Z_{t,\varepsilon}^K(x^{(j)})\big),
  \quad
  \tilde Q_N(\theta):=\tfrac1N\sum_{j=1}^N\log\ell\big(\theta\mid
  x^{(j)},t,Z_t^m(x^{(j)})\big),
\]
with $\hat\theta_t\in\arg\max_\Theta\hat Q_N$ and $\tilde\theta_t\in
\arg\max_\Theta\tilde Q_N$.

\begin{theorem}\label{thm1}
Under Assumptions~\ref{a1}--\ref{a6}, $\hat\theta_t\xrightarrow{p}\theta_t^\star$.
\end{theorem}
\begin{proof}
\emph{(i) Consistency of the infeasible estimator.} Under
Assumptions~\ref{a1}--\ref{a2}, the uniform law of large numbers gives $\sup_{\theta\in\Theta}|\tilde
Q_N(\theta)-Q(\theta)|\xrightarrow{p}0$. Combined with well-separation
(Assumption~\ref{a3}), argmax consistency~\citep{van2000asymptotic}
yields $\tilde\theta_t\xrightarrow{p}\theta_t^\star$.

\emph{(ii) Coupling the feasible and infeasible objectives.} By
Assumption~\ref{a5}, $y\mapsto\log\ell(\theta\mid x,t,y)$ is
$L(x,t)$-Lipschitz uniformly in $\theta$. Applying Lemma~\ref{lem1} at each
$x=x^{(j)}$ with $g_\theta(y):=\log\ell(\theta\mid x^{(j)},t,y)$ (Lipschitz
constant $L(x^{(j)},t)$),
\begin{equation}\label{eq:pointwise}
  \Big|\mathbb E\big[\log\ell(\theta\mid x^{(j)},t,\hat Z_{t,\varepsilon}^K(x^{(j)}))\big]
  -\mathbb E\big[\log\ell(\theta\mid x^{(j)},t,Z_t^m(x^{(j)}))\big]\Big|
  \le L(x^{(j)},t)\cdot L_F\varepsilon(N),
\end{equation}
uniformly in $\theta\in\Theta$. Since the right-hand side does not depend
on $\theta$, taking $\sup_\theta$ inside the absolute value and averaging
over $j$:
\[
  \mathbb E\Big[\sup_{\theta\in\Theta}|\hat Q_N(\theta)-\tilde Q_N(\theta)|\Big]
  \le \tfrac1N\sum_{j=1}^N\mathbb E\Big[\sup_{\theta\in\Theta}
  \big|\log\ell(\theta\mid x^{(j)},t,\hat Z^K_{t,\varepsilon}(x^{(j)}))-\log\ell(\theta\mid x^{(j)},t,Z_t^m(x^{(j)}))\big|\Big] .
\]
By Assumption~\ref{a5}, the summand's integrand is bounded by
$L(x^{(j)},t)|\hat Z^K_{t,\varepsilon}(x^{(j)})-Z_t^m(x^{(j)})|$; applying
Lemma~\ref{lem1} to $g(y)=|y|$ (or directly to the Lipschitz bound in
expectation),
\[
  \mathbb E\Big[\sup_{\theta}|\hat Q_N(\theta)-\tilde Q_N(\theta)|\Big]
  \le L_F\varepsilon(N)\cdot\tfrac1N\sum_{j=1}^N L(x^{(j)},t)
  \xrightarrow{p} L_F\varepsilon(N)\cdot\mathbb E[L(X,t)] \to 0
\]
using $\varepsilon(N)\to0$ (Assumption~\ref{a6}) and $\mathbb E[L]<\infty$
(Assumption~\ref{a5}). Hence,
$\sup_\theta|\hat Q_N(\theta)-\tilde Q_N(\theta)|\xrightarrow{p}0$.

\emph{(iii) Conclusion.} Combining (i) and (ii),
$\sup_\theta|\hat Q_N(\theta)-Q(\theta)|\xrightarrow{p}0$. Well-separation
of $\theta_t^\star$ and argmax consistency give
$\hat\theta_t\xrightarrow{p}\theta_t^\star$.
\end{proof}

\begin{theorem}[Consistency of the CETE estimator]\label{thm-cete}
Under Assumptions~\ref{a1}--\ref{a6} applied to both $t=0,1$, and
continuity of $\theta\mapsto\xi_t(x;\theta)$ for each $x$
(Assumption~\ref{a2}), for $P_X$ almost everywhere $x\in\mathcal X$,
\[
  \hat\tau_{ext}(x)=\hat\xi_1(x)-\hat\xi_0(x)\xrightarrow{p}\tau_{ext}(x) .
\]
\end{theorem}
\begin{proof}
It follows by applying Theorem~\ref{thm1} for $t=0,1$ on each treatment group; combine via the continuous mapping theorem for
$(\theta_1,\theta_0)\mapsto\xi_1(x;\theta_1)-\xi_0(x;\theta_0)$.
\end{proof}
\endgroup

\begin{remark}
Assumption~\ref{ass:tail_unconf} is imposed on the limit $Z_t$ and is used only for identification, the finite-sample $\varepsilon$-max-sampler requires no finite-$m$ analogue, since this theorem shows that its surrogate maxima converge to the correct limit, with finite-$m$ bias controlled by the Lipschitz order.
\end{remark}

\section{Tail Unconfoundedness Example}
\label{sec:tail_examples}
We present a quantitative example where tail
unconfoundedness (Assumption~\ref{ass:tail_unconf}) is satisfied but standard
unconfoundedness (Assumption~\ref{ass:unconf}) fails.
 
Consider the density function $p(y_1\mid x,t)$ following a
Beta distribution $\mathcal B(\alpha,\beta)$ with parameters $\alpha=t+1$
and $\beta=x$ for $x>0$ and $t\in\{0,1\}$. The upper endpoint is $1$, and
near it the density behaves as $(1-y)^{\beta-1}$, so the upper-tail Weibull
index is $\xi(x)=-1/\beta=-1/x$, independent of $t$. Concretely, the
maximum-domain-of-attraction is $\mathrm{GEV}(0,1,\xi(x))$ with
normalizing constants \citep{roncalli2020handbook}

\[
  b_m(x)=1,\qquad a_m(x)=\left(\frac{m\,\Gamma(\alpha+x)}{\Gamma(\alpha)\Gamma(x+1)}\right)^{-1/x}.
\]

As $m\to\infty$, $a_m(x)\to 0$, and the limiting law of
$Z_1^m(x)$ depends only on $x$, so tail unconfoundedness holds. In
contrast, since $\alpha=t+1$ depends on $t$, the full conditional density
$p(y_1\mid x,t)$ depends on $t$, so standard unconfoundedness fails.

\section{Algorithm}
\label{sec:algo}
We describe the algorithm for the CETE estimation in~\autoref{alg}.
\begin{algorithm}[th!]
\SetKwInput{KwInput}{Input}         
\SetKwInput{KwOutput}{Output}       
\SetKwInput{KwData}{Data}
\SetKwFunction{main}{Main}
\SetKwFunction{loss}{Loss}
\SetKwFunction{maxs}{MaxSampler}
\SetKwComment{Comment}{$\triangleright$ }{}
\SetCommentSty{scriptsize}

\DontPrintSemicolon

\KwData{$S = \{\left(x^{(i)},t^{(i)},y^{(i)}\right)\}_{i=1}^{n}$}
\KwInput{A neural network $N_{\theta} = \{(\hat{\mu}_{t},\hat{\sigma}_{t},\hat{\xi}_{t})\}_{t=0}^{1}$ : $s$-Learner, $t$-Learner, TARNet, or CFR}
\KwOutput{A trained neural network $N_{\theta^*}$}

    \SetKwProg{Fn}{Function}{:}{}
    \Fn{\maxs{$S,m$}}{
        Initialize an empty set $S^m$ \;
        \For{$t=0,1$}{
        Let $S_t = \{j \mid t^{(j)} = t\}$ and $n_t = |S_t|$\;
        Number of clusters:\;$K_t = n_t\; \text{div}\; m$ \;
        Run $K$-means on $\{x^{(j)}\}_{j\in S_t}$ to cluster them into $K_t$ clusters.\;
        \For{$i=1,2,\ldots,K_t$}{
        Take $S_{t,i} = \{j \in S_t \mid (x^{(j)},t^{(j)},y^{(j)})\; \text{in the} \; i^{th} \;\text{cluster}\}$\;
        
        Select $s = \arg \underset{j\in S_{t,i}}{\text{max}} y^{(j)}$\;

        Add $(x^{(s)},t^{(s)},y^{(s)})$ to $S^{m}$
        }
        }
        
        \KwRet $S^m$
    }
        

        
    
    \SetKwProg{Fn}{Function}{:}{}
    \Fn{\loss{$S^{m},N_{\theta}$}}{
        For each $(x^{(i)},t^{(i)},y^{(i)}) \in S^m$ the negative conditional log-likelihood is   \\  
        $
        \begin{aligned}
        L_i = & - \log(P(y^{(i)} \mid \mu_0(x^{(i)}),\sigma_0(x^{(i)}),\xi_0(x^{(i)}))) \mathbf{1}_{t^{(i)}=0} \\ 
        & - \log(P(y^{(i)} \mid \mu_1(x^{(i)}),\sigma_1(x^{(i)}),\xi_1(x^{(i)}))) \mathbf{1}_{t^{(i)}=1}
        \end{aligned}
        $
        
        \KwRet $\frac{1}{n}\sum_{i=1}^{n} L_i$
    }
    \SetKwProg{Fn}{Function}{:}{}
    \Fn{\main}{
    Choose $m$\;
    Run the max-sampler\;
    $S^m = \maxs{S,m}$\;
    Train the neural network $N_\theta$ by minimizing the loss: \; L = \loss{$S^m,N_{\theta}$}\;

    \KwRet The trained neural network $N_{\theta*}$
    }

\caption{Learning the Parameters of the Domain of Attraction of Individual Potential Outcomes}
\label{alg}
\end{algorithm}

\section{Datasets Descriptions}
\label{sec:description}
In this section, we present a detailed description of the datasets used in the numerical study.
A summary is provided in Table~\ref{dataset_table}.
\paragraph{IHDP} 
The IHDP dataset was first introduced by~\citet{hill} and it is based on real covariates available from the Infant Health and Development Program (IHDP), which studied the effect of development programs on children. 
The features in this dataset come from a randomized controlled trial. 
The potential outcomes are simulated using Setting B. 
The dataset consists of $747$ individuals (specifically, $139$ in the treatment group and $608$ in the control group), each with $25$ features. 
The potential outcomes are generated according to the following distributions for the control group:
$$Y_{0} = \exp(\beta^{T}\cdot(X + W)) + \eta_0, \;\;  \text{with} \; \eta_0 \sim \mathcal{N}(0,1)$$
and for the treatment group:
$$Y_{1} = \beta^{T}(X+W) - \omega + \eta_1, \;\;  \text{with} \; \eta_1 \sim \mathcal{N}(0,1) $$
where $W$ has the same dimension as $X$ with all entries equal $0.5$ and $\omega=4$. 
The regression $25$ dimensional coefficient vector $\beta$ is randomly sampled from a categorical distribution with the support $(0, 0.1, 0.2, 0.3, 0.4)$ and the respective probabilities $\mu = (0.6, 0.1, 0.1, 0.1,0.1)$. 
We refer to this dataset as IHDP (Original), which corresponds to a Gumbel domain of attraction. 
We also generate two new versions of IHDP by changing the noise variables of the potential outcomes. 
The first version is given by:
$$Y_{0} = \exp(\beta^{T}\cdot(X + W)) + \eta_0, \;\;  \text{with} \; \eta_0 \sim \text{GEV}(0,1,0.1+\|X\|_2^{1/10})$$
and
$$Y_{1} = \beta^{T}(X+W) - \omega + \eta_1, \;\;  \text{with} \; \eta_1 \sim \text{GEV}\left(0,1,\log(1.1+\|X\|_2)\right) $$
We refer to this version as IHDP (Fr\'echet). 
The second version is generated with the following potential outcomes:
$$Y_{0} = \exp(\beta^{T}\cdot(X + W)) + \eta_0, \;\;  \text{with} \; \eta_0 \sim \text{GEV}(0,1,-0.1-\|X\|_2^{1/10})$$
and
$$Y_{1} = \beta^{T}(X+W) - \omega + \eta_1, \;\;  \text{with} \; \eta_1 \sim \text{GEV}\left(0,1,-\log(1.1+\|X\|_2)\right) $$
We refer to this version as IHDP (Weibull).


\begin{table}[t]
\caption{Overview of the causal inference datasets for CETE and ETE Estimation.}
\label{dataset_table}
\begin{center}
\begin{small}
\begin{sc}
\begin{tabular}{l|cc}
\toprule
\multicolumn{1}{l}{\bf Name}  &\multicolumn{1}{c}{\bf Type} &\multicolumn{1}{c}{\bf IPO MDA}\\ 
\midrule
IHDP (original)&Semi-Synthetic &Gumbel  \\
IHDP (Fr\'echet) &Semi-Synthetic & Fréchet  \\
IHDP (Weibull) &Semi-Synthetic & Weibull\\
\bottomrule
\end{tabular}
\end{sc}
\end{small}
\end{center}
\footnotesize{Individual Potential Outcomes Maximum Domain of Attraction (\textbf{IPO MDA}) if consistent across all individuals.}
\end{table}

\section{GEV Likelihoods}
\label{sec:gev_likelihoods}
For reference, we provide the correspondences between the bulk distributions and tail distributions for our simulated experiments. 
\subsubsection{Gumbel Distribution}
The density function of a Gumbel is
$$
f(y ; \mu(x), \beta(x))= \frac{1}{\beta(x)} \exp\left({-\left(\frac{y-\mu(x)}{\beta(x)} + e^{- \frac{y-\mu(x)}{\beta(x)}}\right)}\right)
$$
with $\mu(x)$ a real function for the location parameter and $\beta(x)$ a positive real function of the scale.

Then, the conditional negative log-likelihood is
$$
L(\mu,\beta) = \sum_{i=1}^{n}\left( \frac{y_i-\mu(x_i)}{\beta(x_i)} + e^{\frac{y_i-\mu(x_i)}{\beta(x_i)}} \right)+ n \log{\beta(x_i)}.
$$

\subsubsection{Fr\'echet Distribution}
The density function of a Fr\'echet distribution is
$$
\begin{aligned}
f(y ; \alpha(x), s(x),m(x))
= \frac{\alpha(x)}{s(x)} \left(\frac{y-m(x)}{s(x)}\right)^{-1-\alpha(x)} e^{-\left(\frac{y-m(x)}{s(x)}\right)^{-\alpha(x)}}.
\end{aligned}
$$
Then, the conditional negative log-likelihood is
$$
\begin{aligned}
L(\alpha,s,m) = \sum_{i=1}^{n} \Bigg[(1+\alpha(x_i)) \log\left(\frac{y_i-m(x_i)}{s(x_i)}\right) + \left(\frac{y_i-m(x_i)}{s(x_i)}\right)^{-\alpha(x_i)} + \log\left(\frac{s(x_i)}{\alpha(x_i}\right)\Bigg].
\end{aligned}
$$
\subsubsection{Weibull Distribution}
The density function of a Weibull is
$$
\begin{aligned}
f(y ; \lambda(x), k(x)) =\frac{k(x)}{\lambda(x)}\left(\frac{y}{\lambda(x)}\right)^{k(x)-1} \exp{\left(-(y / \lambda(x))^{k(x)}\right)} \mathbbm{1}_{y \geq 0}
\end{aligned}
$$
with $k(x)$ and $\lambda(x)$ two positive functions of the shape and the scale, respectively. 
The conditional negative log-likelihood is
$$
\begin{aligned}
L(\lambda,k) = \sum_{i=1}^{n} \left(\frac{y_i}{\lambda(x_i)}\right)^{k(x_i)} + k(x_i)\log (\lambda(x_i)) - \log(k(x_i)) + (1-k(x_i)) \log(y_i). 
\end{aligned}
$$

\section{Convergence to GEV Distributions }
\label{sec:convergence}
\subsubsection{Gaussian to Gumbel}
Let $Y(x) \sim \mathcal{N}(\mu(x),\sigma(x))$.
The limiting GEV is given by
$$
\frac{Z_{m}(x)-b_m(x)}{a_m(x)} \underset{m\to + \infty}{\longrightarrow} \text{Gumbel}(0,1)
$$
where 
$$
b_m = \mu(x) + \sigma(x) \left(\sqrt{2 \log(m)} - \frac{\log(\log(m))+\log(4\pi)}{2\sqrt{\log(m)}}\right)
$$
and 
$$
a_m = \frac{\sigma(x)}{\sqrt{2 \log(m)}}.
$$
This implies that for large enough $m$,
$$
Z_m(x) \approx c_m(x) g + d_m(x)
$$
with $g \sim \text{Gumbel}(0,1)$.

\subsubsection{Log-Gamma to Fr\'echet}
Let $Y(x) \sim \mathcal{LG}(\alpha(x),\beta(x))$. 
The limiting GEV is given by
$$
\frac{Z_{m}(x)-b_m(x)}{a_m(x)} \underset{m\to + \infty}{\longrightarrow} \text{GEV}(0,1,\beta(x)^{-1})
$$
where 
$$
a_m(x) = 0 \; \text{and} \;
b_m(x) = \frac{\left(m \log(m)^{\alpha(x) -1}\right)^{1/\beta(x)}}{\Gamma(\alpha(x))}.
$$
Then,
$$
Z_m(x) \approx a_m(x) \; \text{GEV}(0,1,\xi(x)) + b_m(x)
$$
with 
$$
\xi(x) = \beta(x)^{-1}.
$$

\subsection{Beta to Weibull}
 
Let $Y(x)\sim\mathcal B(\alpha(x),\beta(x))$. We have that
\[
  \frac{Z_m(x)-b_m(x)}{a_m(x)}\xrightarrow[m\to\infty]{d}
  \mathrm{GEV}\!\left(0,1,-\beta(x)^{-1}\right)
\]
where
\[
  b_m(x)=1,\qquad
  a_m(x)=\left(\frac{m\,\Gamma(\alpha(x)+\beta(x))}
    {\Gamma(\alpha(x))\Gamma(\beta(x)+1)}\right)^{-1/\beta(x)}.
\]
Then,
\[
  Z_m(x)\approx a_m(x)\,\mathrm{GEV}(0,1,\xi(x))+b_m(x),\qquad
  \xi(x)=-\beta(x)^{-1}.
\]
 
\begin{figure*}[t!]
    \centering
    \begin{subfigure}[t]{0.3\textwidth}
    \centering 
    \includegraphics[width=\textwidth]{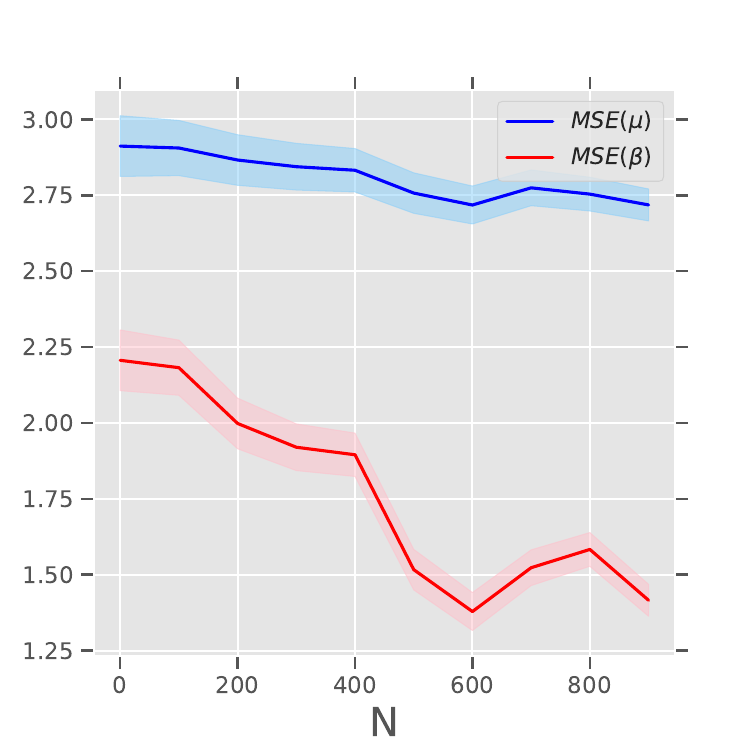}
    \caption{Gumbel, $D=1$}
    \end{subfigure}\hfill
    \begin{subfigure}[t]{0.3\textwidth}
    \centering 
    \includegraphics[width=\textwidth]{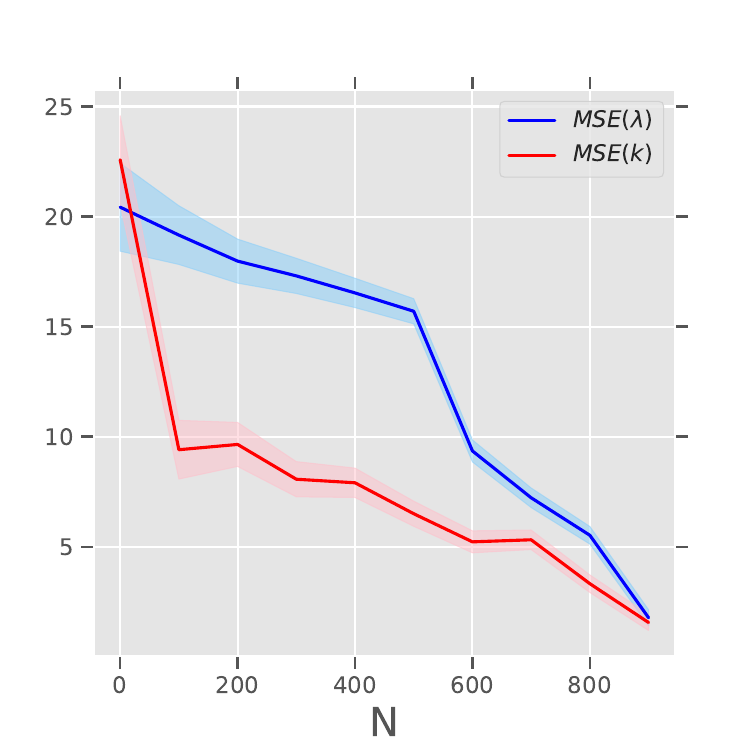}
    \caption{Weibull, $D=1$}
    \end{subfigure}\hfill
    \begin{subfigure}[t]{0.3\textwidth}
    \centering 
    \includegraphics[width=\textwidth]{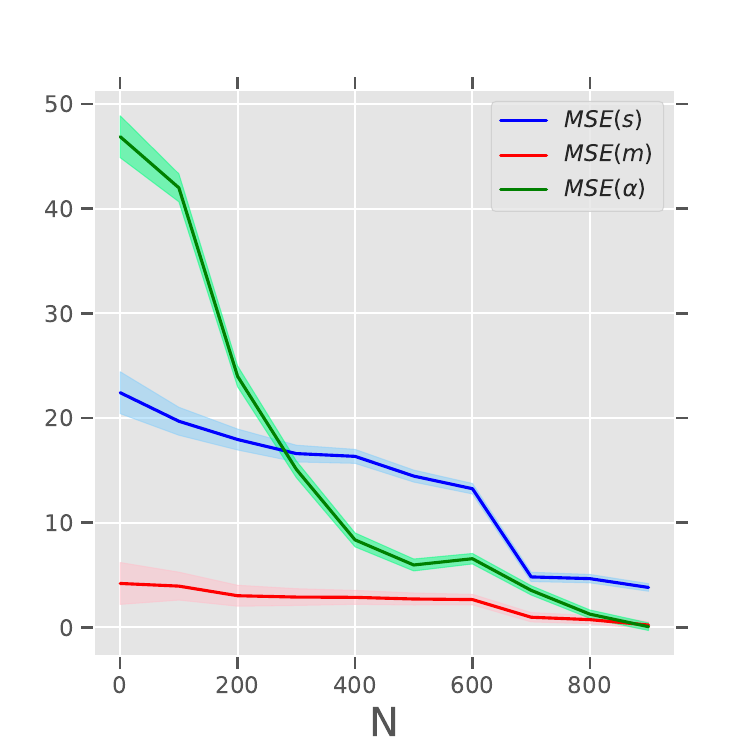}
    \caption{Frechet, $D=1$}
    \end{subfigure}
    \centering
    \begin{subfigure}[t]{0.3\textwidth}
    \centering 
    \includegraphics[width=\textwidth]{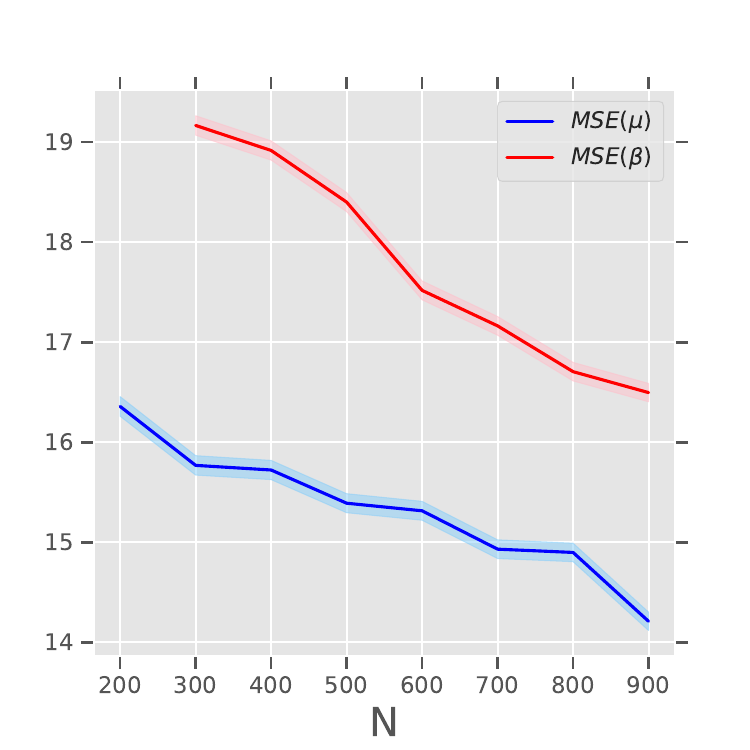}
    \caption{Gumbel, $D=5$}
    \end{subfigure}\hfill
    \begin{subfigure}[t]{0.3\textwidth}
    \centering 
    \includegraphics[width=\textwidth]{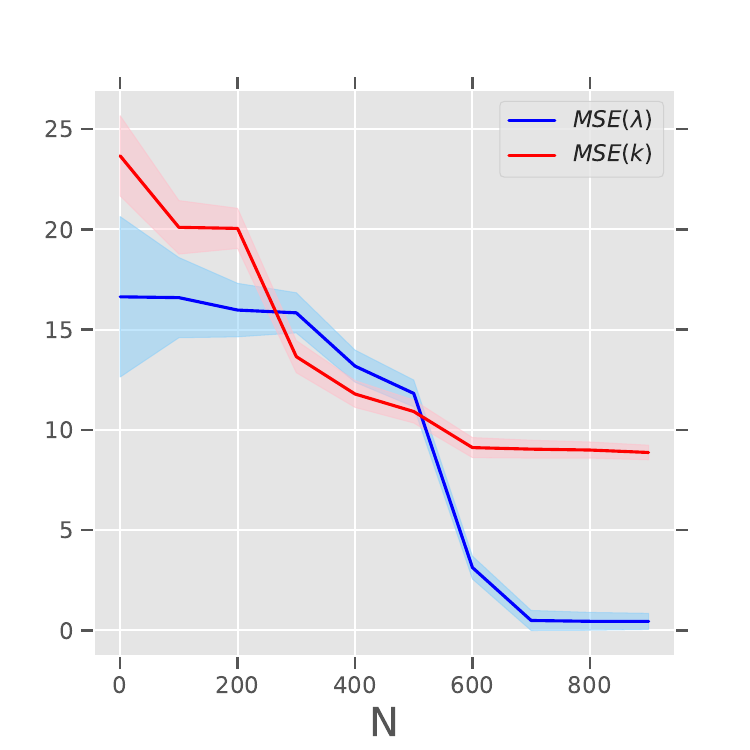}
    \caption{Weibull, $D=5$}
    \end{subfigure}\hfill
    \begin{subfigure}[t]{0.3\textwidth}
    \centering 
    \includegraphics[width=\textwidth]{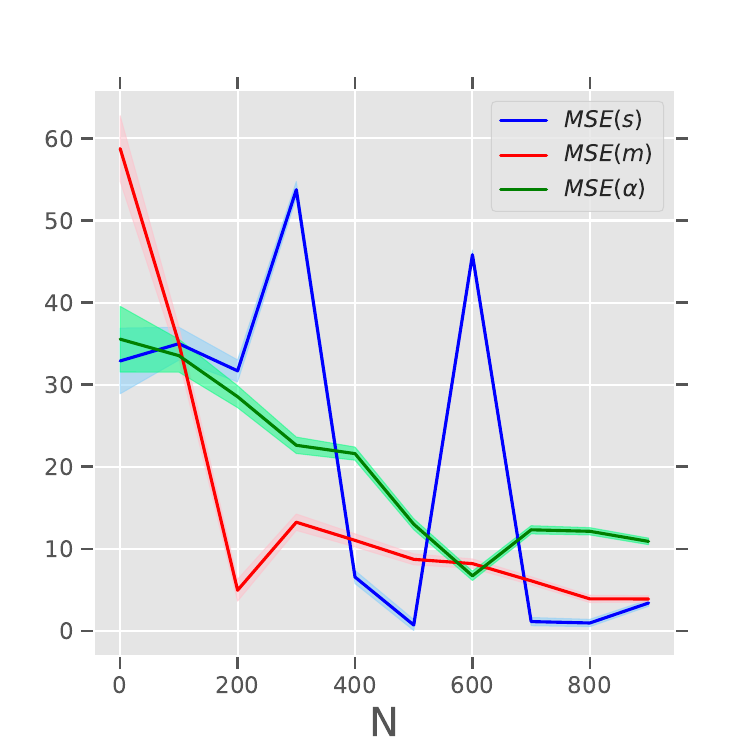}
    \caption{Frechet, $D=5$}
    \end{subfigure}
    \centering
    \begin{subfigure}[t]{0.3\textwidth}
    \centering 
    \includegraphics[width=\textwidth]{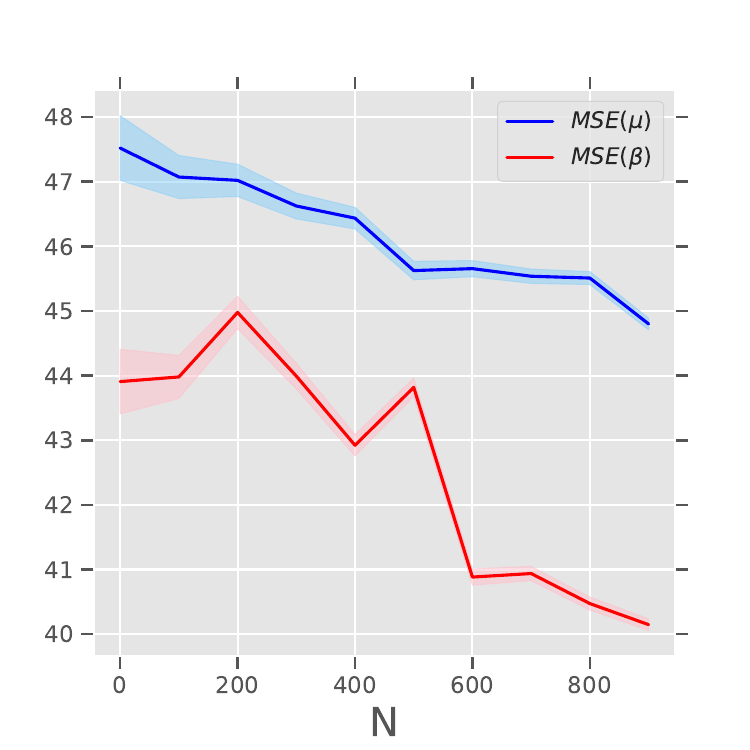}
    \caption{Gumbel, $D=10$}
    \end{subfigure}\hfill
    \begin{subfigure}[t]{0.3\textwidth}
    \centering 
    \includegraphics[width=\textwidth]{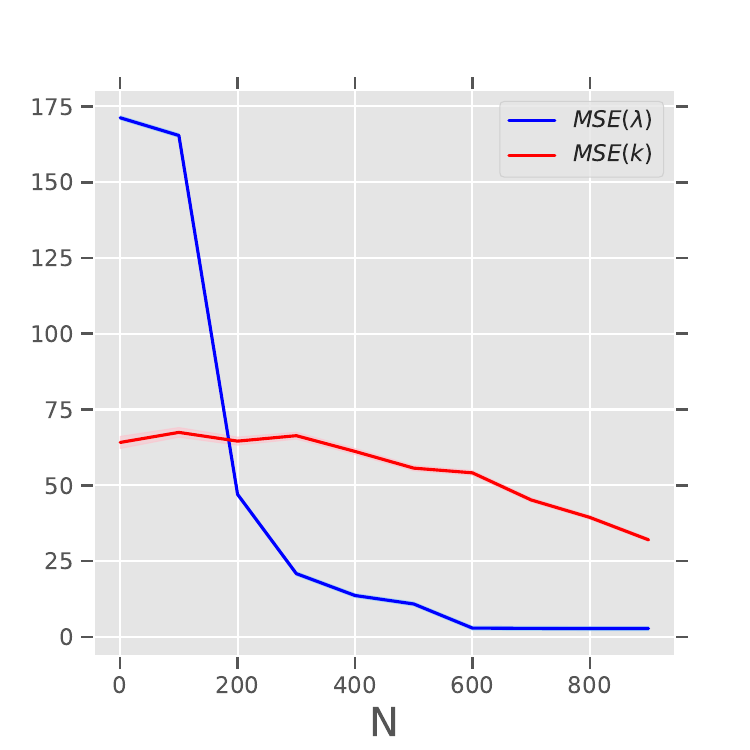}
    \caption{Weibull, $D=10$}
    \end{subfigure}\hfill
    \begin{subfigure}[t]{0.3\textwidth}
    \centering 
    \includegraphics[width=\textwidth]{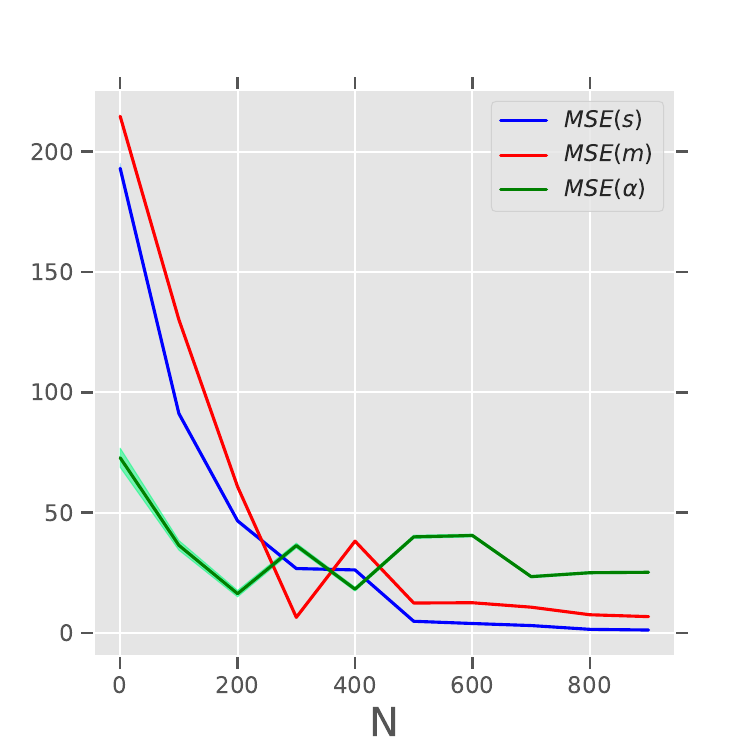}
    \caption{Frechet, $D=10$}
    \end{subfigure}
    \caption{Experiments For Conditional GEV Performance as a function of $N$ and $D$}
\label{fig:ablation_cgev}
\end{figure*}
\begin{figure}[t!]
    \centering
    \begin{subfigure}[t]{0.3\textwidth}
    \centering 
    \includegraphics[width=\textwidth]{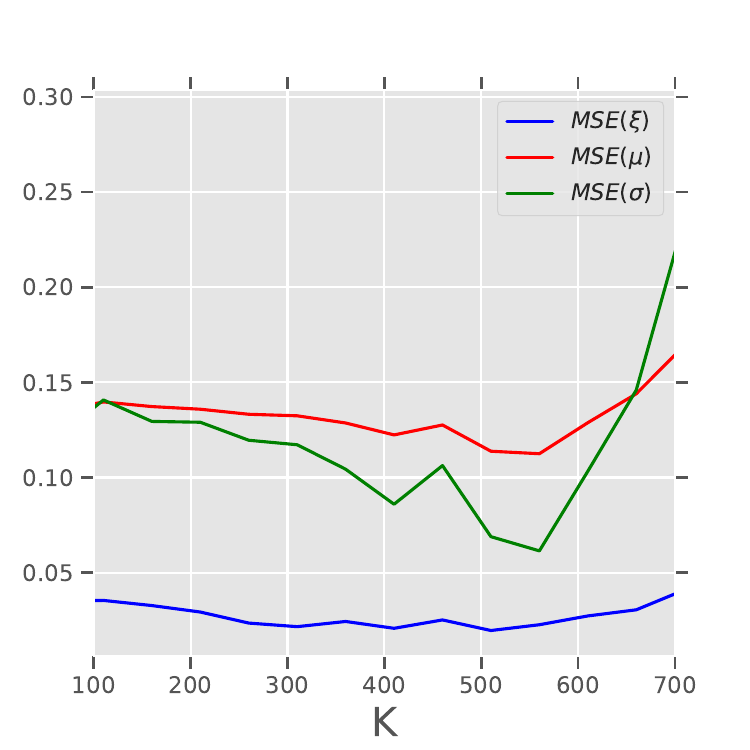}
    \caption{$D=1$}
    \end{subfigure}\hfill
    \begin{subfigure}[t]{0.3\textwidth}
    \centering 
    \includegraphics[width=\textwidth]{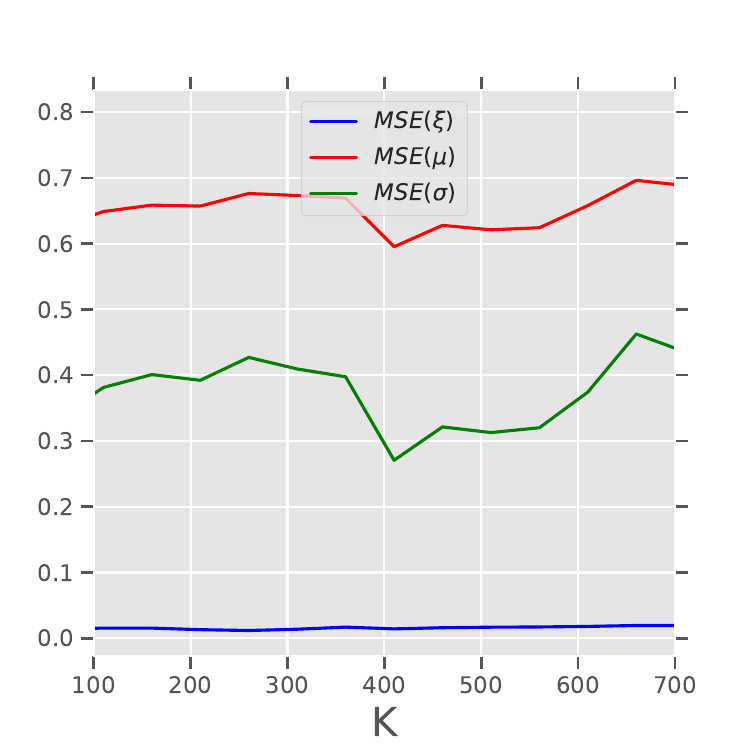}
    \caption{$D=5$}
    \end{subfigure}\hfill
    \begin{subfigure}[t]{0.3\textwidth}
    \centering 
    \includegraphics[width=\textwidth]{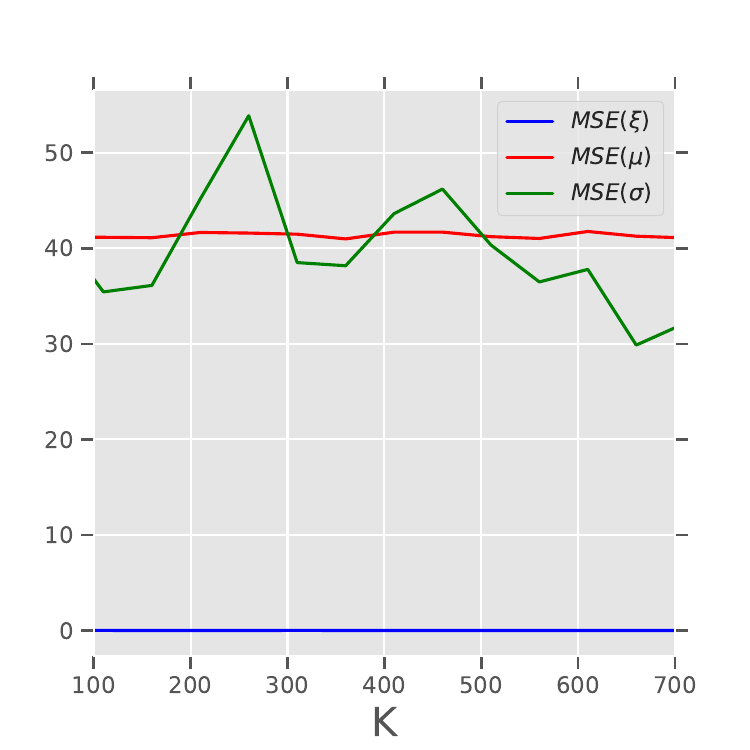}
    \caption{$D=10$}
    \end{subfigure}
    \caption{Experiments for the Gaussian to Gumbel case as a function of K}
\label{fig:ablation_exp1}
\end{figure}

\section{Additional Experiments}
\label{sec:ablation}
In this section, we provide additional experiments to validate the proposed method.
We start with verifying our methods for estimating conditional GEVs and the estimated scale parameter of MDA for the Log-Gamma distribution. 
Then, an ablation is performed on the number of the data size $N$ and the dimensionality $D$. 
Finally, we present an ablation on the number of clusters $K$. 
\subsubsection{Estimation from GEV observations}
\begin{figure*}
    \centering
    \begin{subfigure}[t]{0.3\textwidth}
    \centering 
    \includegraphics[width=\textwidth]{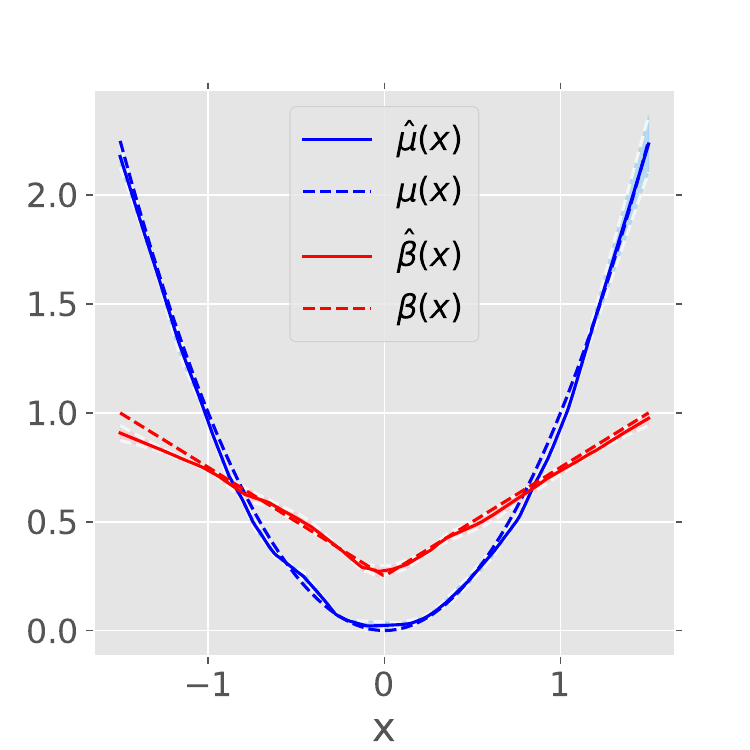}
    \caption{Gumbel distribution with the true location and scale functions $\mu(x) = x^2$ and $\beta(x) = 0.25+ \frac{|x|}{2}$.}
    \label{fig:exp0_gumbel}
    \end{subfigure}\hfill
    \begin{subfigure}[t]{0.3\textwidth}
    \centering 
    \includegraphics[width=\textwidth]{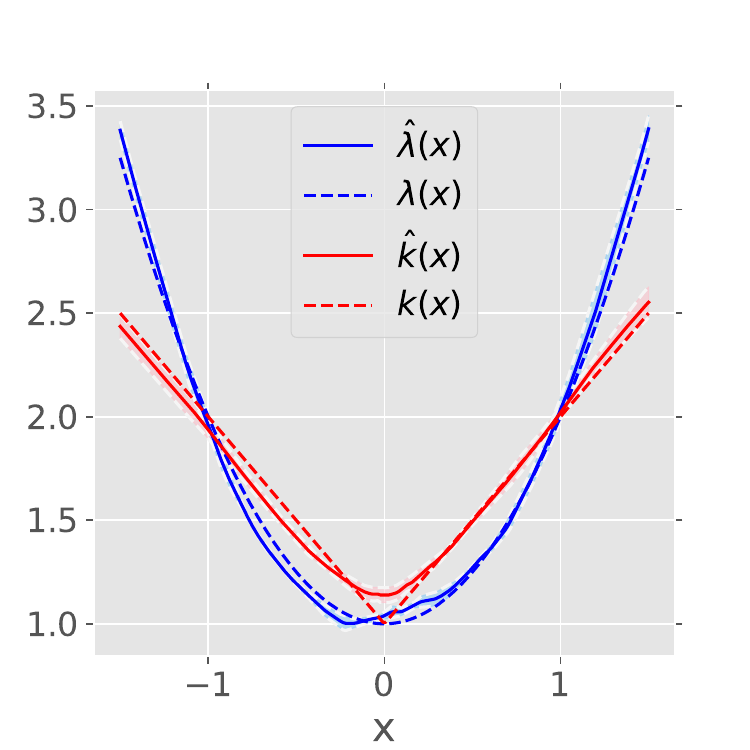}
    \caption{Weibull distribution with the true scale and shape functions $\lambda(x) = 1+x^2$ and $\beta(x) = 1+|x|$.}
    \label{fig:exp0_weibull}
    \end{subfigure}\hfill
    \begin{subfigure}[t]{0.3\textwidth}
    \centering 
    \includegraphics[width=\textwidth]{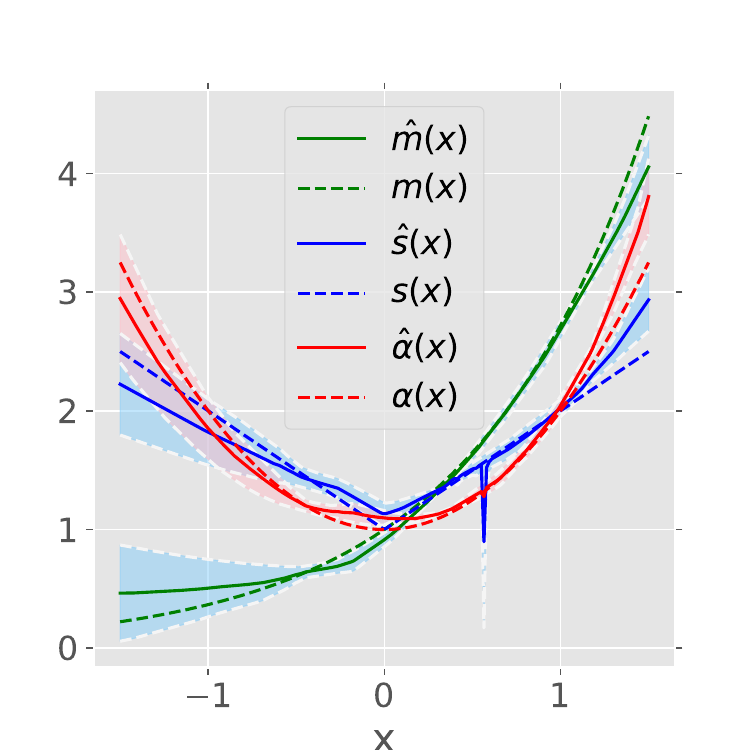}
    \caption{Fr\'echet distribution with the true location and scale functions $m(x) = e^{x}$, $s(x) = 1+|x|$, and $\alpha(x) = 1 + x^{2}$.}
    \label{fig:exp0_frechet}
    \end{subfigure}
    \caption{Estimation of distribution parameters from GEV observations for Gumbel, Weibull, and Fr\'echet distributions.}
    \label{fig:exp0}
\end{figure*}
We estimate the functional parameters of the GEV given a dataset $S = \{(x^{(i)},y^{(i)})\}_{i=1}^n$ where $P(Y \mid X=x)$ follows a GEV distribution. 
We consider the cases of the Weibull, Gumbel, and Fr\'echet distributions by learning the parameters of each of the distributions.
The full likelihood functions for each case were detailed in Appendix~\ref{sec:gev_likelihoods}.
A comparison of the estimated and true functions is given in~\autoref{fig:exp0}.

\paragraph{Conditional Gumbel}
We assume that the data are generated according to
$
Y(x) \sim \text{Gumbel}(\mu(x),\beta(x))
$
with $\mu(x)=x^2$ and $\beta(x) = 0.25 + \frac{|x|}{2}$.
Figure~\ref{fig:exp0_gumbel} shows the performance of our model on this dataset. 


 \begin{figure*}[t!]
    \centering
    \begin{subfigure}[t]{0.3\textwidth}
    \centering 
    \includegraphics[width=\textwidth]{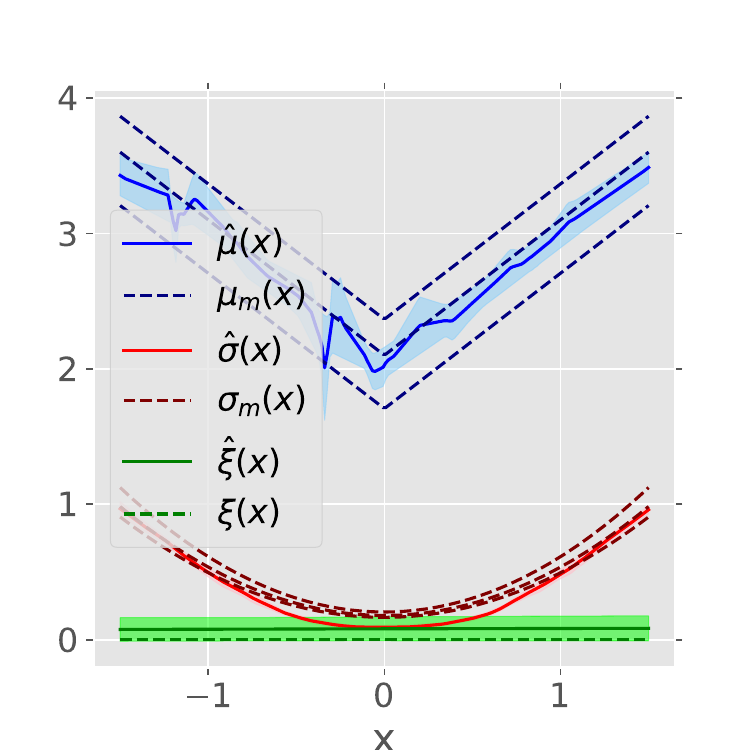}
    \caption{Estimation of parameters of the limiting Gumbel distribution generated by a Gaussian with $\mu(x) = |x|$ and $\sigma(x) = x^2$. }
    \label{fig:exp1-gaussian}
    \end{subfigure}\hfill
    \begin{subfigure}[t]{0.3\textwidth}
    \centering 
    \includegraphics[width=\textwidth]{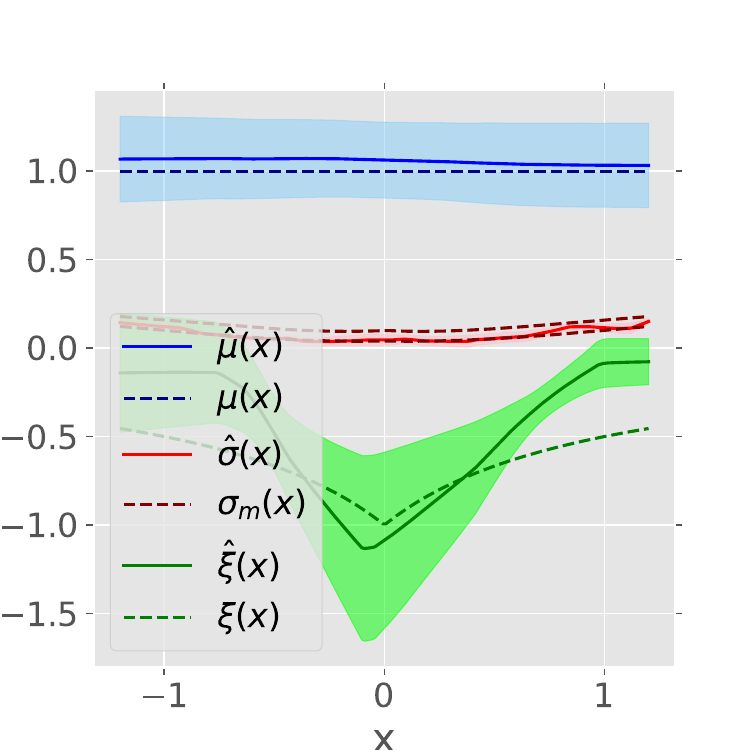}
    \caption{Estimation of parameters of limiting Weibull generated by a beta distribution with the parameters given by $\alpha(x) = |x|+1$ and $\beta(x) = x^2 +1$.}
    \label{fig:exp1-beta}
    \end{subfigure}\hfill
    \begin{subfigure}[t]{0.3\textwidth}
    \centering 
    \includegraphics[width=\textwidth]{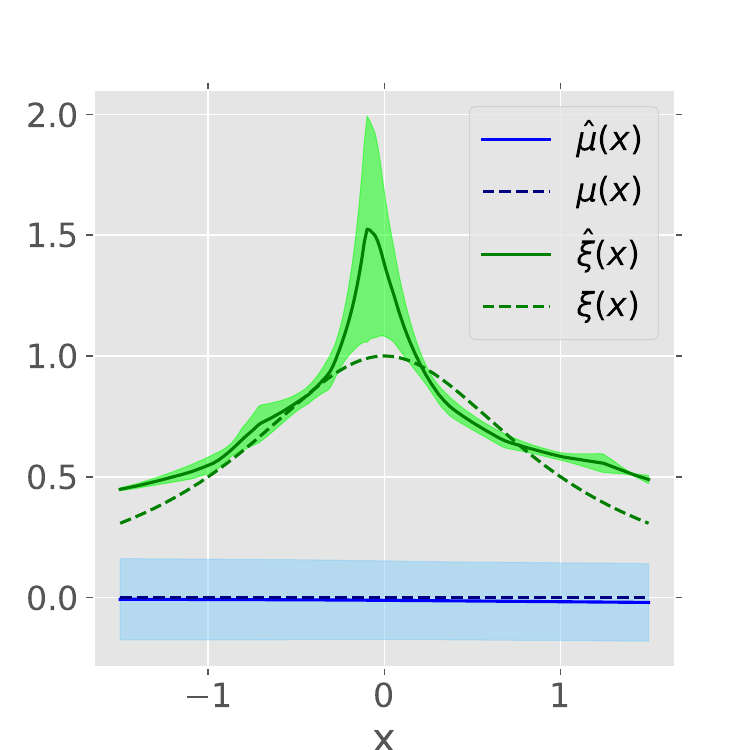}
    \caption{Estimation of parameters of limiting Fr\'echet distribution generated by a log-gamma distribution with parameters given by $\alpha(x) = |x|+1$ and $\beta(x) = x^2 +1$.}
    \label{fig:exp1-log_gamma}
    \end{subfigure}
    \caption{Estimation of limiting GEV distribution parameters from observations of Gaussian, beta, and log-gamma distributions. All parameters are estimated by taking the $\max$ of the $K$-means partition of the feature space.}
    \label{fig:exp1}
\end{figure*}

\paragraph{Conditional Weibull}
We assume that the data are generated according to
$
Y(x) \sim \text{Weibull}(\lambda(x),k(x))
$
with $\lambda(x) = 1+x^2$ and $\beta(x) = 1 + |x|$.
Figure~\ref{fig:exp0_weibull} illustrates the results of fitting the estimator to the conditional data. 

\paragraph{Conditional Fr\'echet}
We assume that the data are generated according to:
$
Y(x) \sim \text{Fréchet}(\alpha(x),s(x),m(x))
$
where $\alpha(x) = 1 + x^2$, $s(x) = 1 + |x|$, and $m(x) = \exp(x)$. 
Figure~\ref{fig:exp0_frechet} illustrates the results for fitting in the Fr\'echet case.  

These experiments validate the estimation procedure in the main text when observing the extreme data. 
Next, we consider the case where we observe samples of the data and need to use the $\varepsilon$-max-sampler to estimate the extreme observations.
We do this by applying the $\varepsilon$-max-sampler to data from a log-gamma distribution and estimate the corresponding Fre\'chet distribution. 
The results are presented in~\autoref{fig:exp1-log_gamma_sigma} for estimating the scale $\sigma$ parameter.

We also estimated the scale parameter for the Log-Gamma distribution in Figure~\ref{fig:exp1-log_gamma_sigma}. 
\begin{figure}[t!]
    \centering
    \includegraphics[width=4cm]{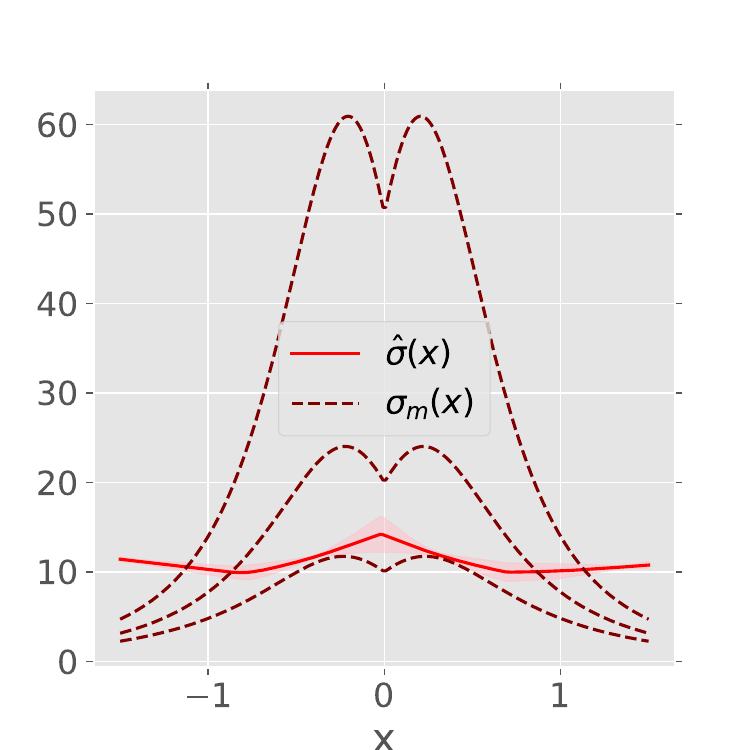}
    \caption{Estimation of the limiting scale parameter $\sigma$ of the maximum of the log-gamma distribution.}
    \label{fig:exp1-log_gamma_sigma}
\end{figure}

\subsubsection{Estimating conditional limiting GEV distributions}
For these experiments, we consider the $\varepsilon$-max-sampler with a total of $2000$ random samples of the limiting GEV.
In Appendix~\ref{sec:ablation}, we include detailed experiments on fitting conditional GEV distributions when they are directly observed. We generate data based on a distribution with a known MDA, and then compare the recovered GEV distribution parameters to sequences known for their convergence to the true parameters.
We provide full details of the convergence of these distributions to their corresponding GEV distributions in Appendix~\ref{sec:convergence}.
The results are presented in~\autoref{fig:exp1} where we compare the true functions to the estimated ones for the different datasets.
We additionally provide an ablation study on the impact of varying the block sizes and the dimensionality of the data in Appendix~\ref{sec:ablation}.

\textbf{Gaussian distribution: Gumbel MDA.}
We generate $X \sim \mathcal{N}(0,1)$, and $Y(x) \sim \mathcal{N}(\mu(x),\sigma(x))$, where $\mu(x) = |x|$ and $\sigma(x) = x^{2}$. The corresponding GEV is the Gumbel distribution. Moreover, we estimate the limiting conditional GEV parameters with $\hat{\mu}(x), \hat{\sigma}(x), \hat{\xi}(x)$.
Figure~\ref{fig:exp1-gaussian} illustrates this result where we compare a sequence converging to the ground truth to the estimated parameters. 

\textbf{Beta distribution: Weibull MDA.} 
We generate samples $X\sim \mathcal{N}(0,1)$ and $Y(x) \sim \mathcal{B}(\alpha(x),\beta(x))$ with $\alpha(x) = |x| +1$ and $\beta(x) = x^2 +1$.
The corresponding GEV is a Weibull distribution. We illustrate the estimates of these parameters in Figure~\ref{fig:exp1-beta}.

\textbf{Log-Gamma distribution: Fr\'echet MDA.} 
We generate samples $X\sim \mathcal{N}(0,1)$ and $Y(x) \sim \mathcal{LG}(\alpha(x),\beta(x))$ with $\alpha(x) = |x| +1$ and $\beta(x) = x^2 +1$ 
where $\mathcal{LG}$ is the log-gamma distribution. The corresponding GEV is a Fr\'echet distribution. Figure~\ref{fig:exp1-log_gamma} illustrates the performance of the proposed method. 
The results for the scale parameter are presented in Appendix~\ref{sec:ablation}

\paragraph{Ablation on $N$ and $D$}
We study how the performance of our method changes when varying the data size $N$ and the dimensionality $D$, as illustrated in Figure \ref{fig:ablation_cgev}.
The ablation suggests that the Weibull and Frechet cases scale when observing large $N$ but the Gumbel case is more difficult to estimate.

\paragraph{Ablation on number of clusters}
We study the performance of our method when varying the number of clusters $K$ and the dimensionality $D$ for a fixed number of data points $N=10^{4}$. 
The results are illustrated in Figure~\ref{fig:ablation_exp1}.

\section{Joint Estimation of the Location and Shape Parameters}
\label{app:joint_mu_xi}

To overcome some of the numerical instabilities, we demonstrate that our framework extends to jointly estimating the conditional location $\mu(x)$ and shape $\xi(x)$, with the scale fixed at $\sigma = 1$. The key observation is that, for known
$\sigma$, the GEV location is identified from the conditional median through the quantile (return-level) formula
\citep[Eq.~(3.10)]{coles2001introduction}: the median of
$\mathrm{GEV}(\mu, \sigma, \xi)$ is
\begin{equation}
\label{eq:gev_median}
m = \mu + \sigma\,\frac{(\ln 2)^{-\xi} - 1}{\xi}.
\end{equation}
We therefore proceed in three stages: (i)~estimate the conditional
median $m(x)$ by neural quantile regression, which is well-posed under
heavy tails as it requires no moment conditions; (ii)~holding
$\hat{m}(x)$ fixed, estimate $\xi(x)$ by maximizing a penalized GEV
likelihood, in which observations violating the support constraint
$1 + \xi (y - \mu)/\sigma > 0$ are excluded from the likelihood and
penalized linearly (a plain clamp of the support term would make
violations \emph{reward} large $\xi$, creating a spurious
$\xi \to \infty$ direction); and (iii)~recover
$\hat{\mu}(x)$ by inverting Eq.~\eqref{eq:gev_median} at $\hat{\xi}(x)$
and refit $\hat{\xi}(x)$ once against the corrected location.
We illustrate this on a synthetic dataset with $X \sim \mathcal{N}(0,1)$,
$\mu(x) = x^2$, $\xi(x) = \log(1.1 + x^2)$, and
$Y = \mu(X) + \eta$, $\eta \sim \mathrm{GEV}(0, 1, \xi(X))$, with
$n = 10{,}000$ observations. As shown in
Figure~\ref{fig:joint_mu_xi}, both functions can be recovered accurately
(location: RMSE $0.060$, correlation $0.999$; shape: RMSE $0.101$,
correlation $0.980$ over the grid $x \in [-2, 2]$).

\begin{figure}[t]
\centering
\includegraphics[width=0.95\linewidth]{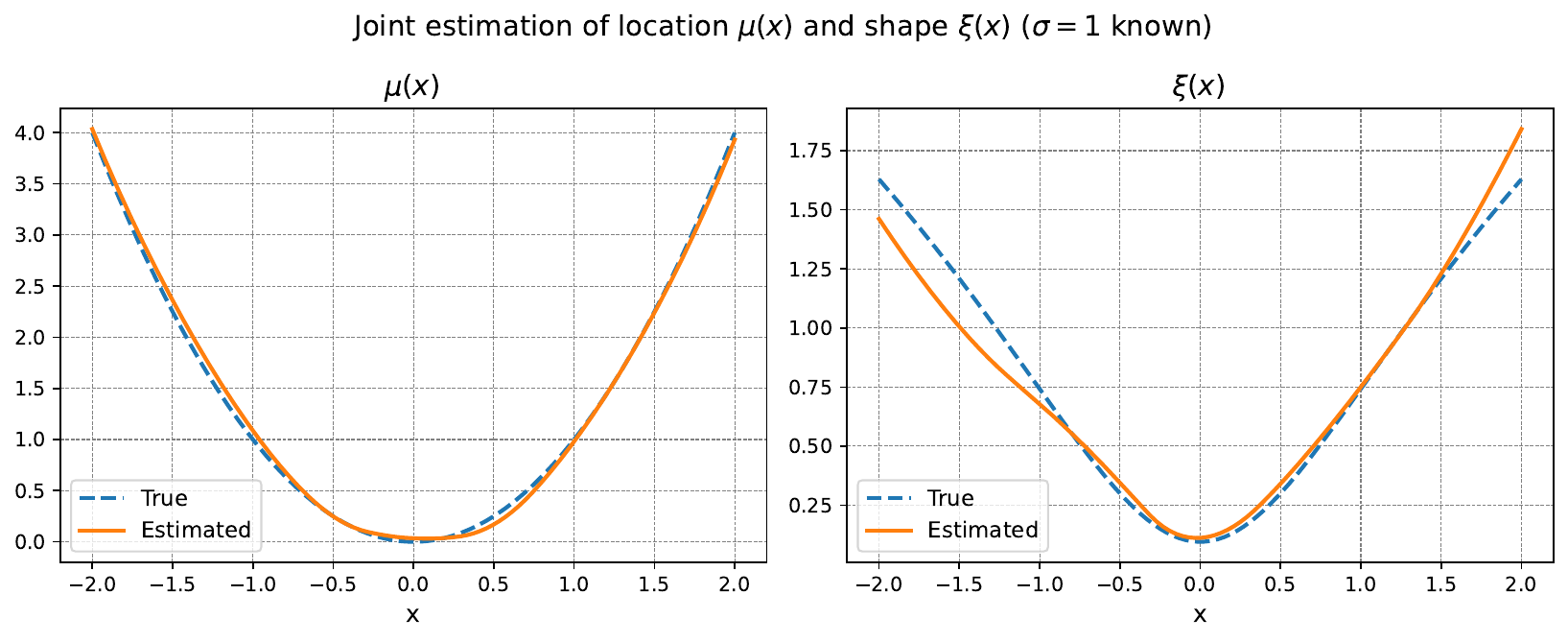}
\caption{Joint estimation of the conditional location $\mu(x)$ (left)
and shape $\xi(x)$ (right) with $\sigma = 1$ known, using the
median-anchored staged procedure of
Appendix~\ref{app:joint_mu_xi}.}
\label{fig:joint_mu_xi}
\end{figure}

\end{document}